\documentclass[aps,prl,twocolumn,superscriptaddress]{revtex4-2}

\usepackage[T1]{fontenc}
\usepackage{graphicx,subfigure}
\usepackage{amsmath,amssymb,bm,amsthm}
\usepackage{mathtools}
\usepackage{multirow}
\usepackage{textcomp}
\usepackage{float}
\usepackage[dvipsnames]{xcolor}
\usepackage[normalem]{ulem}
\usepackage{dsfont}
\usepackage{url}
\usepackage{mathrsfs}
\usepackage{bbm} 
\usepackage{lipsum}
\usepackage{tikz}
\usepackage{enumerate}
\usepackage[colorlinks=true, allcolors=blue]{hyperref}

\newcommand{\ket}[1]{\left| #1 \right\rangle}
\newcommand{\bra}[1]{\left\langle #1 \right|}

\providecommand{\kb}[1]{|#1\rangle\!\langle #1 |}
\providecommand{\tr}{\mathrm{tr}}
\providecommand{\ddf}[2]{\frac{\mathrm{d}#1}{\mathrm{d}#2}}

\providecommand{\supp}{\mathrm{supp}}
\providecommand{\RPetz}{\RRP_{\sigma,\EE}}
\providecommand{\RRP}{\RR^{\mathrm{P}}}
\providecommand{\RP}{\hat{R}^{\mathrm{P}}}
\providecommand{\RR}{\mathcal{R}}
\providecommand{\PP}{\mathcal{P}}
\providecommand{\LL}{\mathcal{L}}
\providecommand{\HH}{\mathcal{H}}
\providecommand{\KK}{\mathcal{K}}
\providecommand{\EE}{\mathcal{E}}
\providecommand{\II}{\mathcal{I}}
\providecommand{\VV}{\mathcal{V}}
\providecommand{\XX}{\mathcal{X}}
\providecommand{\YY}{\mathcal{Y}}
\providecommand{\Pos}{\mathrm{Pos}}
\providecommand{\Proj}{\hat{\Pi}}

\newtheorem{theorem}{Theorem}
\newtheorem{corollary}{Corollary}

\begin{document}

\title{Optimality Condition for the Petz Map}

\author{Bikun Li}
\email{bikunli@uchicago.edu}
\affiliation{Pritzker School of Molecular Engineering, University of Chicago, Chicago, Illinois 60637, USA}
\author{Zhaoyou Wang}
\affiliation{Pritzker School of Molecular Engineering, University of Chicago, Chicago, Illinois 60637, USA}
\author{Guo Zheng}
\affiliation{Pritzker School of Molecular Engineering, University of Chicago, Chicago, Illinois 60637, USA}
\author{Yat Wong}
\affiliation{Pritzker School of Molecular Engineering, University of Chicago, Chicago, Illinois 60637, USA}
\author{Liang Jiang}
\email{liangjiang@uchicago.edu}
\affiliation{Pritzker School of Molecular Engineering, University of Chicago, Chicago, Illinois 60637, USA}

\newcommand{\revise}[1]{{\textcolor{blue!70!green}{#1}}}

\begin{abstract}
In quantum error correction, the Petz map serves as a perfect recovery map when the Knill-Laflamme conditions are satisfied. Notably, while perfect recovery is generally infeasible for most quantum channels of finite dimension, the Petz map remains a versatile tool with near-optimal performance in recovering quantum states. This work introduces and proves, for the first time, the necessary and sufficient conditions for the optimality of the Petz map in terms of entanglement fidelity. In some special cases, the violation of this condition can be easily characterized by a simple commutator that can be efficiently computed. We provide multiple examples that substantiate our new findings.  
\end{abstract}

\maketitle

Protecting fragile quantum information from noisy environments is a critical challenge in quantum information processing, as the quantum advantage over classical systems largely relies on maintaining nontrivial quantum states with high fidelity. Over the past two decades, quantum error correction (QEC) has emerged as an essential tool for shielding quantum systems against noise~\cite{ShorCode, GKP2001, lidar2013QEC}. Its applications span various physical platforms and promise architectures that extend beyond fault-tolerant quantum computers~\cite{Gottesman1998TheoryofFTQC,Grimsmo2021PRX_QECGKP}. In addition, QEC has established itself as a cornerstone in other fields such as reliable quantum communication~\cite{Bennett1996PRA,Quantum_repeaters_2023} and high-precision quantum sensing~\cite{Reiter2017, Rojkov2022PRL}. The fundamental principle of quantum error-correcting codes involves encoding logical quantum states with redundancies. After transmission through a noise channel, a decoding or recovery operation is performed to suppress errors, preserving the coherence of the logical information. The conception, benchmarking, and implementation of concrete QEC schemes have become focal points of research in quantum information science in recent years~\cite{Xu2024, Ni2023, Sivak2023, deNeeve2022, Postler2022,Cai2024}. 

A foundational work regarding QEC is the Knill-Laflamme~(KL) conditions~\cite{KLCondition1997}, which mathematically present the necessary and sufficient conditions for faithful recovery of encoded quantum information. Under the KL conditions, it is well-known that the \textit{Petz map} can perfectly recover the encoded quantum information~\cite{Barnum_Knill_2002}. However, the KL conditions are commonly violated in most finite quantum codes and physical noise channels. In these cases, the optimal recovery map can be found using numerical approaches~\cite{Reimpell_Werner_2005PRL, Fletcher2007PRA}. Essentially, since the Petz map approximates the reverse of a quantum channel by the analog of Bayes' theorem, it still guarantees near-optimal performance when perfect quantum state restoration is impossible~\cite{Barnum_Knill_2002, Ng_AQECcond_2010PRA, zheng2024nearoptimal}. The approach of Petz map has been applied to various areas involving the reconstruction of quantum information, such as approximate quantum error correction~\cite{Ng_AQECcond_2010PRA}, many-body open quantum systems~\cite{AWWW2018PRA,Aw2021_AVS,Hu2024_PRB,Kwon2019PRX,Kwon2022PRL,sang2024,Parzygnat2023axiomsretrodiction}, and black hole physics~\cite{Chen2020, Cree2022, Nakayama2023}, due to its versatility and many other appealing features.

This work investigates the sufficient and necessary conditions for the Petz map to be the optimal recovery map for QEC, even when the KL conditions are violated. These conditions enable efficient verification of optimality without the need to run any optimization algorithm. 
We highlight that the optimality of the Petz map has been studied for different objective functions~\cite{Iten2017TIT, bai2024quantumbayesrulepetz}.  
Our work mainly focuses on the entanglement fidelity, which generalizes some results from \cite{Iten2017TIT} with a different route of proof.
Moreover, we provide many interesting quantum channels that satisfy our criterion. In conjunction with recent work~\cite{zheng2024nearoptimal}, our findings make analytical estimation of the optimal entanglement fidelity for these quantum channels possible.

\begin{figure}[t]
    \centering
    \includegraphics[width=0.7\linewidth]{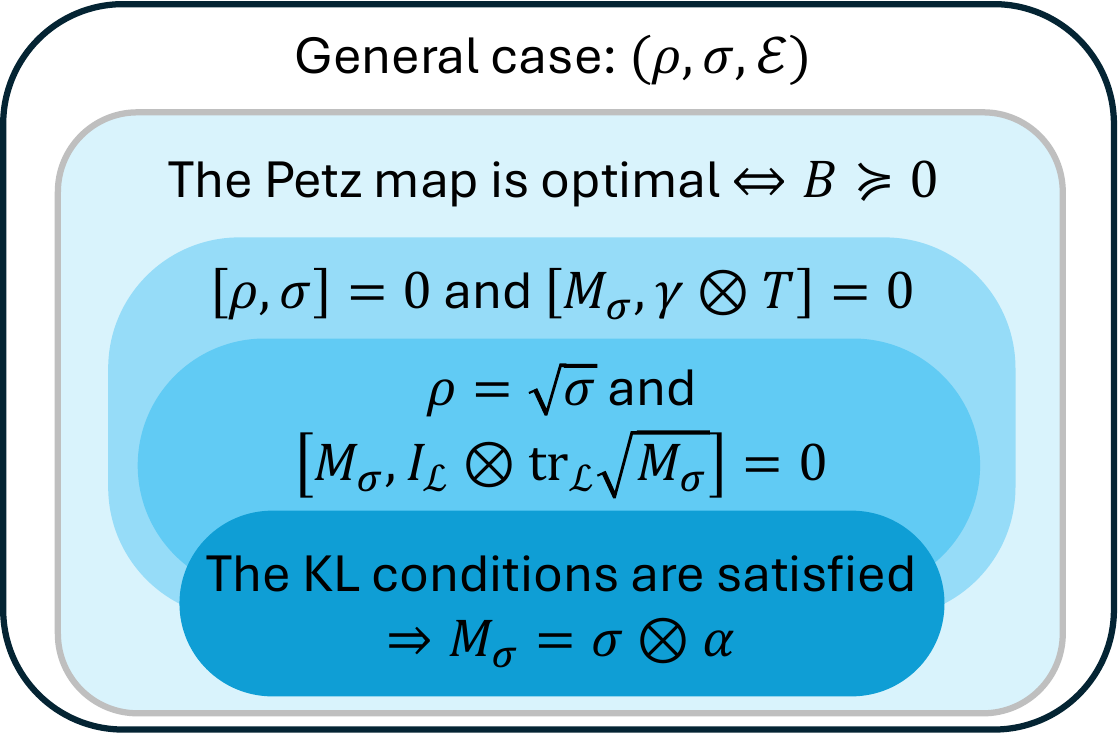}
    \caption{This Venn diagram illustrates the generalization of the Knill-Laflamme conditions, in the sense of the Petz map optimizing $F_e$. Our result applies to the general case with input state $\rho$, the reference state $\sigma$, and the noise channel $\EE$. The blue regions show that the Petz map achieves optimality under different assumptions.}
    \label{fig:venn}
\end{figure}

\textit{Preliminary.} --- 
Generally, the QEC process involves encoding a $d$-dimensional "logical" Hilbert space $\LL:=\mathbb{C}^d$ into an $n$-dimensional "physical" space $\HH:=\mathbb{C}^n$. The encoded state residing in $\HH$ subsequently experiences noise, which may corrupt the original quantum information. In this work, we denote this \textit{composite} quantum operation by $\EE$, which has the Kraus representation $\mathcal{E}\sim\{\hat{E}_k\}$ with Kraus operators $\hat{E}_k:\LL\to\HH$ for $1\le k\le n_K$. We call $\EE$ a quantum channel, but it can be generalized to the trace nonincreasing case.
$\{\hat{E}_k\}$ induces an $n_K$-dimensional Hilbert space $\KK:=\mathbb{C}^{n_K}$, which is considered as the "environment" required to define the complementary channel of $\EE$.

Given a Hilbert space $\VV$, $\Pos(\VV)$ denotes the space of all positive semidefinite (PSD) operators acting on $\VV$. The notation $I_{\VV}$ represents the identity matrix (or projector) acting on $\VV$. If $\VV$ corresponds to a subsystem, then $\tr_{\VV}(\cdots)$ represents the partial trace over the subspace $\VV$.

We define $M_\sigma\in\Pos(\LL \otimes \KK)$ as the QEC matrix associated with the reference state $\sigma\in\Pos(\LL)$ and the noise channel $\EE$, with matrix components given by
\begin{equation}\label{eq:QECmat_ref}
    (M_\sigma)_{[\mu k],[\nu \ell]}:=\bra{\mu}\sqrt{\sigma}\hat{E}_k^\dagger\hat{E}_\ell\sqrt{\sigma}\ket{\nu}\;.
\end{equation}
Here, $\ket{\mu},\ket{\nu}$ denote orthonormal basis vectors in $\LL$. We refer specifically to the special case $M:=M_{I_\LL}$ as the QEC matrix, which has been demonstrated to be a powerful tool in prior works~\cite{Albert2018PRA_QECMatrix,zheng2024nearoptimal,zheng2024performanceachievableratesGKP}.
The generalized definition in Eq.~\eqref{eq:QECmat_ref} becomes crucial in our subsequent discussion of the Petz map with the reference state $\sigma$.
Note that $M^T$ can also be interpreted as the Choi matrix of  $\EE^{\mathrm{c}}$, which is the complementary channel of $\EE$~\cite{Beny2010AQEC, SuppMat}. 

In this work, an unknown logical input state is represented by a general mixed state $\rho\in\Pos(\LL)$.
To recover $\rho$ from the noise corruption, the error correction process is followed by a recovery map $\RR\sim\{\hat{R}_i\}$, where $\hat{R}_i:\HH\to \LL$ for $i = 1,2,\cdots$. 
Maintaining the entanglement of $\rho$ with other subsystems is crucial in QEC for practical applications. Therefore, in this study, the entanglement fidelity serves as the performance metric~\cite{Schumacher1996} for the composite map $\RR\circ \EE$, defined as 
\begin{equation}\label{eq:entfidelity} 
    \begin{aligned}
    F_{e}(\rho, \RR\circ \EE)&:=\bra{\psi_\rho}\II\otimes \left(\RR\circ \EE\right)(\psi_\rho)\ket{\psi_\rho}\\&=\sum_{i,k}\left|\tr \big(\hat{R}_i\hat{E}_k\rho\big) \right|^2\;,
    \end{aligned} 
\end{equation} 
where $\ket{\psi_\rho} \in \LL \otimes \LL$ is the state vector representing the canonical purification of $\rho$, $\psi_\rho \equiv \kb{\psi_\rho}$, and $\II$ is the identity map acting on the reference subsystem $\LL$. 
When $\rho$ is a maximally mixed state in $\LL$, $F_{e}$ is closely related to the average fidelity~\cite{PhysRevA.60.1888}, which can be conveniently estimated via quantum tomography. Additionally, $\min_\rho F_e$ bounds the diamond distance between $\RR\circ\EE$ and $\II$ by the Fuchs–van de Graaf inequalities~\cite{PhysRevA.71.062310}.

The optimality of a recovery map is usually based on the definition of the objective function, which may differ across contexts. 
In this work, a recovery map $\RR$ is called \textit{optimal} if it achieves $F_{e} = F^{\mathrm{op}}_{e} := \max_{\RR} F_{e}(\rho,\RR\circ \EE)$ for a given $\EE$. An optimal recovery map $\RR$ is further termed \textit{perfect} if it achieves $F_{e}=1$. This perfection occurs if and only if the KL conditions~\cite{KLCondition1997} are satisfied:
\begin{equation}\label{eq:KLcondition}
    M=I_\LL\otimes \alpha\;,
\end{equation}
where $\alpha\in \Pos(\KK)$ satisfies $\tr(\alpha) = 1$.
Eq.~\eqref{eq:KLcondition} implies that $M_\sigma = \sigma \otimes \alpha$.
In this situation, a universal solution for perfect recovery is the Petz map $\RPetz$ of the \textit{reference state} $\sigma$~\cite{Petz1988,Barnum_Knill_2002}. Specifically, $\RPetz$ is a completely positive, trace nonincreasing map with Kraus operator 
\begin{equation}\label{eq:Petz_def}
    (\hat{R}^{\mathrm{P}}_{\sigma,\EE})_k := \sigma^{\frac{1}{2}}\hat{E}_k^\dagger \EE(\sigma)^{-\frac{1}{2}}\;,
\end{equation}
where negative powers of operators or matrices are defined via the pseudoinverse throughout this work. 
The role of $\sigma$ is to ensure that when $\sigma = \rho$, we have $\RRP_{\rho,\EE} \circ \EE(\rho) = \rho$, although this condition does not necessarily optimize $F_e$.
Nonetheless, the Petz map guarantees near-optimal performance with $(F^{\mathrm{op}}_e)^2 \le F_e(\rho,\RRP_{\rho,\EE}\circ \EE) \le F^{\mathrm{op}}_e$~\cite{Barnum_Knill_2002}.
When the reference state $\sigma$ in the Petz map need not coincide with the input state $\rho$, typically it is necessary to have $\supp(\rho)\subseteq\supp(\sigma)$ to maximize $\tr(\RPetz\circ\EE(\rho))$.
If the reference state $\sigma$ is a maximally mixed state that has the largest support, then the Petz map becomes the \textit{transpose channel} (TC), denoted as $\RR^{\mathrm{TC}}$~\cite{Ng_AQECcond_2010PRA}.  
If the input state $\rho$ is a maximally mixed state, the entanglement fidelity reduces to the \textit{channel fidelity}. Remarkably, when both $\rho$ and $\sigma$ are maximally mixed states, the channel fidelity given by TC has a simple form~\cite{zheng2024nearoptimal}:
\begin{equation}\label{eq:F_e_TC}
    F^{\mathrm{TC}}_e = \frac{1}{d^2}\left\|\tr_\LL\sqrt{M}\right\|^2_F\;,
\end{equation}
where $\|\cdots\|_F$ denotes the Frobenius norm.

\textit{Main result.} ---
By leveraging the definition of the QEC matrix, we arrive at a compact expression for $F_e$ (Theorem~\ref{thm:ent_fid_petz}), where the recovery map is the Petz map with an arbitrary reference state. All detailed proofs of theorems and corollaries in this work are provided in Supplemental Material~\cite{SuppMat}.
\begin{theorem} \label{thm:ent_fid_petz}
    Let $\rho, \sigma\in\Pos(\LL)$, the entanglement fidelity associated with the input state $\rho$, and the quantum channel $\Phi=\RPetz\circ\EE$ has a compact form given by
    \begin{equation}\label{eq:ent_fid_petz}
        F_e(\rho,\Phi) = \left\|\tr_{\LL}\left(M_\sigma^{-\frac{1}{2}}\left(\sqrt{\sigma}\otimes I_{\KK} \right)M\left(\rho\otimes I_{\KK}\right)\right)\right\|^2_F\;.
    \end{equation}
\end{theorem}
We emphasize that Theorem~\ref{thm:ent_fid_petz} does not impose any matrix support requirements on $\rho$ and $\sigma$. Eq.~\eqref{eq:ent_fid_petz} directly generalizes the result from the previous work, Eq.~\eqref{eq:F_e_TC}, which specifically assumes $\rho=\sigma=I_\LL/d$. 
For the general case with $\supp(\rho)\subseteq \supp(\sigma)$, we can define
\begin{equation}\label{eq:T_def}
    T:=\tr_\LL\left(\left(\gamma^\dagger \otimes I_\KK\right)\sqrt{M_\sigma}\right)\;,
\end{equation}
with $\gamma := \sigma^{-\frac{1}{2}}\rho$, and reduce Eq.~\eqref{eq:ent_fid_petz} to $F_e = \|T\|_F^2$.
If the KL conditions (Eq.~\eqref{eq:KLcondition}) are further applied, the simplified $F_e$ equals $1$ regardless of the choice of $\sigma$ and $\rho$.

Although no known closed-form expression exists for the optimal recovery map when the KL conditions are violated, we can still identify the complete family of triplets $(\rho, \sigma, \EE)$ for which the Petz map is precisely the optimal recovery map. Our result is summarized in Theorem~\ref{thm:condition} (also see Fig.~\ref{fig:venn}).
\begin{theorem}\label{thm:condition}
    Let $\rho, \sigma\in\Pos(\LL)$ with $\supp(\rho) \subseteq \supp(\sigma)$, and let $M_\sigma$ be the QEC matrix associated with the reference state $\sigma$ and the noise channel $\EE$. 
    Then, $\RR = \RPetz$ optimizes  $F_e(\rho, \RR\circ\EE)$ if and only if 
    \begin{equation}\label{eq:B_def}
        B:=\sqrt{M_\sigma}\left(\gamma\otimes T\right)\succeq 0\;.    
    \end{equation}
\end{theorem}

\textit{Proof.} -- 
We provide the outline of the proof as follows. First, we demonstrate the \textit{necessity} of $B\succeq 0$ when $\RPetz$ is optimal.
This is achieved by considering variations of the recovery map $\RR = \RR(\varepsilon)$, defined via its Kraus operators:
\begin{equation}\label{eq:Rexpand}
    \hat{R}_{k}(\varepsilon):= \big(\RP_{\sigma,\EE}\big)_k + \sum_{s=1}^\infty\varepsilon^s \hat{V}^{(s)}_k\;,
\end{equation}
where $\hat{V}^{(s)}_k$ can be carefully and flexibly chosen such that for all sufficiently small $\varepsilon$, $\hat{R}_k(\varepsilon)$ has a  bounded $2$-norm for each $k$, and $\RR(\varepsilon)$ is TP on $\supp(\EE(\sigma))$.
Assuming that $F_e(\rho, \RR(\varepsilon)\circ \EE)$ achieves its global maximum at $\varepsilon = 0$, we must have $\mathrm{d}F_e/\mathrm{d}\varepsilon = 0$ and $\mathrm{d}^2F_e/\mathrm{d}\varepsilon^2 \le 0$. From the expression of these two derivative conditions, it naturally follows that $B\succeq 0$, thereby proving the necessity. 

The \textit{sufficiency} of Theorem~\ref{thm:condition} is proved by examining the optimal conditions of the semidefinite programming problem. Briefly, finding the optimal quantum channel for recovery is equivalent to solving a convex optimization problem:
\begin{equation}
    \begin{array}{lll}
        &\text{minimize}\quad-F_{e}[C_\RR]  \\
        &\text{subject to}\quad G:=I_\HH - \tr_\LL C_\RR \succeq 0 \\
        &\qquad\qquad\quad\;\text{and}\quad C_\RR\succeq 0\;,\\
    \end{array}
\end{equation}
The objective function is the negative entanglement fidelity $F_e(\rho,\RR\circ\EE)$, which can be rewritten as a linear function $F_{e}[C_\RR]=\tr\left(C_\RR^T(\rho^T\otimes I_\HH)C_\EE (\rho^T\otimes I_\HH) \right)$ with respect to the Choi matrix $C_\RR$. 
Here, the components for the Choi matrices of $\RR$ and $\EE$ are given by $(C_{\RR})^{ab}_{\mu\nu}:=\bra{\mu}\RR\left(\ket{a}\!\bra{b}\right)\ket{\nu}$ and $(C_{\EE})^{ab}_{\mu\nu}:=\bra{a}\EE\left(\ket{\mu_{}}\!\bra{\nu_{}}\right)\ket{b}$, respectively, where $\ket{a}$ and $\ket{b}$ are orthonormal basis in $\HH$.
This work verifies the Karush–Kuhn–Tucker (KKT) conditions~\cite{Coutts2021certifying,SuppMat} by investigating the explicit solution of the Lagrange function:
\begin{equation}\label{eq:Lagrange}
    \mathscr{L}(C_\RR,\Lambda,\Gamma):= -F_e[C_\RR]- \tr(\Lambda^TG) - \tr(\Gamma^T C_\RR)\;.
\end{equation}
We find that the Lagrange multiplier solutions are given by
    \begin{align}
        \Lambda &= \tr_{\LL} \big(C_*^T\nabla_{C^T_{\RR}}F_e\big) \;, \label{eq:Lambda_soln}
        \\
        \Gamma &= I_{\LL}\otimes \Lambda - \nabla_{C_\RR^T}F_e\;, \label{eq:Gamma_soln}
    \end{align}
which are evaluated at the optimal Choi matrix $C_* = C_{\RPetz}$.
The challenge here is that for a general PSD $C_\RR$, we do not have $\Lambda\in \Pos(\HH)$ and $\Gamma \in \Pos(\LL\otimes\HH)$, which are the conditions of dual feasibility required by the KKT conditions. 
Surprisingly, $B\succeq 0$ enables $\Lambda$ and $\Gamma$ to be PSD when the recovery map is the Petz map. Together with other KKT conditions, which can be verified easily due to the linearity of $\mathscr{L}$~\cite{SuppMat}, the sufficiency of Theorem~\ref{thm:condition} is proved. 
~$\Box$

If the KL conditions are satisfied, it is easy to deduce that $B=\rho\otimes \alpha \succeq 0$, because $M$ is a tensor product matrix as displayed in Eq.~\eqref{eq:KLcondition}.
Generally, the condition $B\succeq 0$ in Theorem~\ref{thm:condition}  appears obscure and lacks physical meaning. Nevertheless, this new result can be better comprehended under some simple assumptions regarding $\rho$ and $\sigma$.
\begin{corollary}\label{cor:rhosigma_comm}
    Let $\rho, \sigma\in\Pos(\LL)$ satisfying $[\rho,\sigma]=0$ with $\supp(\rho) \subseteq \supp(\sigma)$.
    Then, $\RR = \RPetz$ optimizes  $F_e(\rho, \RR\circ\EE)$ if and only if
    \begin{equation}\label{eq:commutator_eta}
        \left[M_\sigma, \gamma\otimes T\right] = 0\;.
    \end{equation}
\end{corollary}
The above commutation relation is a generalization of the optimality condition of \textit{singlet fraction} in prior work~\cite{Iten2017TIT}, where the Petz map is referred to as the ``pretty good map''.
Specifically, their singlet fraction is defined by the overlap between the bipartite state $\II\otimes (\RR\circ\EE)(\psi_\sigma)\in \Pos(\LL\otimes \LL)$ and $\ket{\psi_{I_\LL}}$, such that the objective function is proportional to $F_e(\sqrt{\sigma},\RR\circ\EE)$. 
Its optimality condition at $\RR = \RPetz$ is immediately recovered by setting $\rho = \sqrt{\sigma}$ in Eq.~\eqref{eq:commutator_eta}:
\begin{equation}\label{eq:commutator_iten}
    \left[M_\sigma, I_\LL\otimes \tr_\LL \sqrt{M_\sigma}\right]=0\;.
\end{equation}

Notice that $B\succeq 0$ is the most general condition, while Eq.~\eqref{eq:commutator_eta} is its special version with commuting $\rho$ and $\sigma$, and Eq.~\eqref{eq:commutator_iten} is a more specialized version which only optimizes $F_e$ for the input $\sqrt{\sigma}$.
We emphasize that there is no equivalence among these three conditions, making this work a meaningful generalization of previous research. 
This conclusion is substantiated with explicit examples in~\cite{SuppMat}, which satisfy either of the first two conditions but violate their special case versions.  

In practice, one may be interested in the structure of $\EE$ that enables the Petz map to maximize $F_e$. The following corollary illustrates one possible underlying structure.
\begin{corollary}\label{cor:M_sigma}
    Let $M_\sigma$ be the QEC matrix associated with the reference state $\sigma$ and the noise channel $\EE$, and $\KK = \bigoplus_s\KK_s$. If $\sqrt{M_\sigma} = \bigoplus_s \beta_s$, where $\beta_s\in\Pos(\LL\otimes \KK_s)$ satisfying $\tr_\LL \beta_s \propto I_{\KK_s}$, then $\RR = \RPetz$ optimizes $F_e(\sqrt{\sigma},\RR\circ\EE)$.
\end{corollary}
A particularly useful scenario of Corollary~\ref{cor:M_sigma} occurs when $\dim \KK_s = 1$ for all $s$. A channel $\EE$ that satisfies the condition $\sqrt{\sigma} \hat{E}^\dagger_k \hat{E}_{\ell} \sqrt{\sigma} = 0$ for all $k \ne \ell$ falls within this case. This implies that the operators $\hat{E}_k$ corresponding to different values of $k$ map $\supp(\sigma)$ to orthogonal subspaces in $\HH$. We emphasize, however, that logical quantum information may still be distorted within each subspace, thereby preventing perfect recovery of logical quantum information. The optimal recovery operation $\RPetz$ in this situation can be realized by perfectly discriminating the state from the ensemble $\{(I_\LL\otimes \hat{E}_\ell)\ket{\psi_\sigma}\}$ with $I_\LL\otimes \sqrt{\sigma}\hat{E}_k^\dagger$, and then leaving the logical space unchanged.
Since $\RP_k$ depends solely on $\hat{E}_k$, the experimental implementation of $\RPetz$ may be straightforward if applying $\hat{E}_k^\dagger$ is simple.

Another notable case of Corollary~\ref{cor:M_sigma} is going to another extreme, where $\KK = \KK_1$, so trivially, there is $\tr_\LL\sqrt{M_\sigma}\propto I_\KK$. 
Examples of channels $\EE$ falling into this category include the completely depolarizing channel and certain dephasing channels when $\sigma$ is a maximally mixed state. In such cases, the Petz map $\RPetz=\RR^{\mathrm{TC}}=\EE^\dagger$ achieves the same performance as the identity channel. More analytical examples are demonstrated in~\cite{SuppMat}.

\begin{figure}[t]
    \centering
    \includegraphics[width=0.95\linewidth]{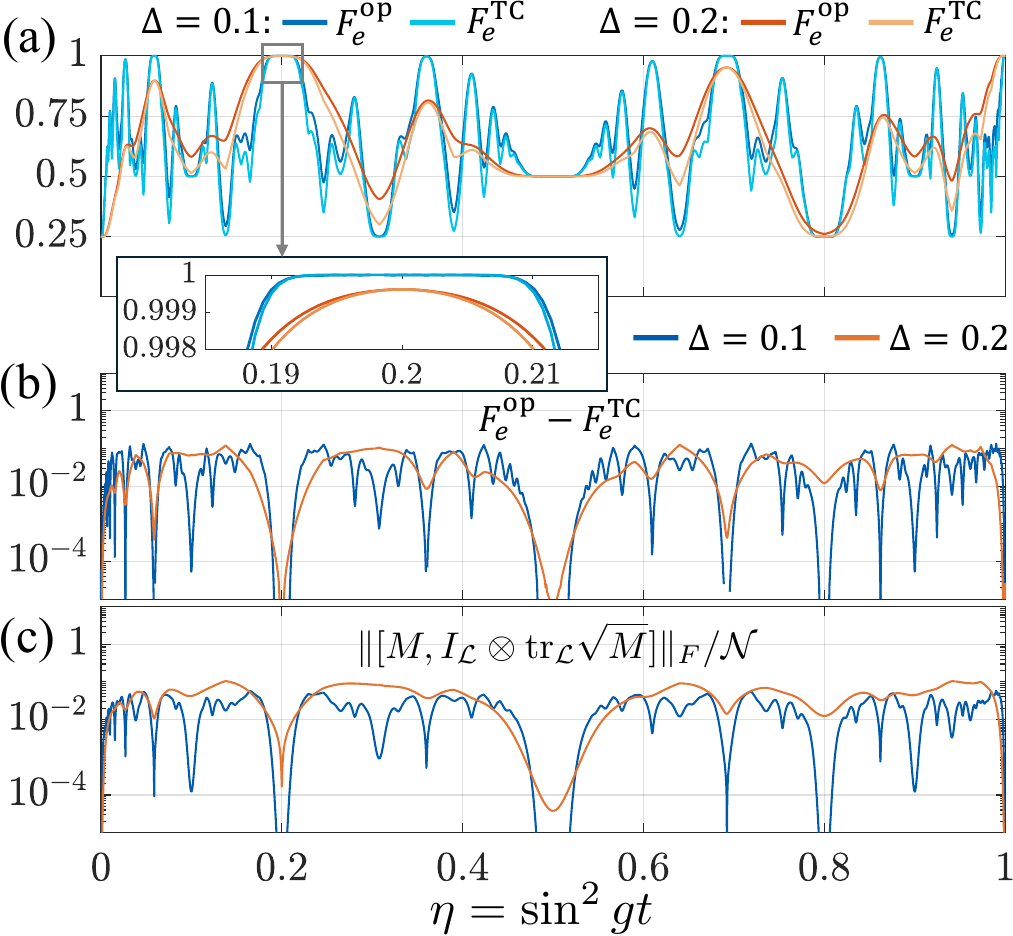}
    \caption{(a) $F^{\mathrm{op}}_e$ and $F^{\mathrm{TC}}_e$ at different values of $\Delta$ and $\eta$. The enlarged inset view displays the negligible gap between $F^{\mathrm{op}}_e$ and $F^{\mathrm{TC}}_e$.  (b) The difference $F^{\mathrm{op}}_e - F^{\mathrm{TC}}_e$ extracted from (a). (c) The normalized Frobenius norm of the commutator in Eq.~\eqref{eq:commutator_iten}, with $\sigma = I_\LL/d$.}
    \label{fig:curves_gkp}
\end{figure}

\textit{Applications.} ---
Theorem~\ref{thm:condition} provides a new, efficient way of quantifying the suboptimality of the Petz map. 
An intuitive way to measure this suboptimality numerically is to evaluate the fidelity gap between the optimal fidelity and that achieved by the Petz map. Typically, the optimal fidelity is obtained by inputting $C_{\EE}$ into iterative algorithms~\cite{Reimpell_Werner_2005PRL} or semidefinite programming~\cite{Fletcher2007PRA}. The first approach typically has a better polynomial complexity $\mathcal{O}(n_{\mathrm{iter}}n^3d^3)$ for $n_{\mathrm{iter}}$ steps. The complexity may be further reduced to $\mathcal{O}(n_{\mathrm{iter}} n_K^3)$ if a small-sized $M$ with $n_K < nd$ is directly available.
In such cases, an effective Choi matrix $C_{\EE}'$, with redundant degrees of freedom removed, can be constructed from $M$ and subsequently fed into the numerical solver.
Numerical experiments suggest that a potential drawback of these iterative approaches is that $n_{\mathrm{iter}}$ may depend on $n$, $n_K$, and $d$, even if the initial guess for the recovery map is $\RPetz$. 
If one's goal is merely to verify the optimality of $\RPetz$, then confirming the positivity of $B$ has complexity $\mathcal{O}(n_K^3)$. In particular, when $[\rho,\sigma] = 0$, computing the norm of the commutator in Eq.~\eqref{eq:commutator_eta} is also more efficient, with complexity $\mathcal{O}(n_K^3)$ independent of the unknown iteration number $n_{\mathrm{iter}}$.

When QEC is considered, the input state of interest is typically $\rho = I_\LL/d$. This is because practical quantum computational tasks usually involve highly entangled logical states that locally appear highly mixed. 
In contrast to previous approximate QEC studies, such as~\cite{Ng_AQECcond_2010PRA}, which focus primarily on qubit stabilizer codes, our goal here is to demonstrate a case where the fidelity gap $F^{\mathrm{op}}_e - F^{\mathrm{TC}}_e$ is negligible compared to the optimal infidelity (for the demonstration purpose, we adopt the TC, where $\sigma = I_\LL/d$).
Specifically, we investigate the quantum transduction model in~\cite{wang2024transduction}, where quantum information is transferred from one engineered bosonic mode to another. 
Let $\ket{\mu_\Delta}$ $(\mu \in \{0,1\})$ be the basis of a finite energy square-lattice Gottesman-Kitaev-Preskill (GKP) code~\cite{GKP2001} for a logical qubit, where $\Delta$ characterizes the average energy by $\approx\hbar\omega/(2\Delta^2)$ with $\hbar\omega$ being the quantized energy per particle.
The intersystem interaction of interest is defined by a beam-splitter-like Hamiltonian $\hat{H} = -ig(\hat{a}^\dagger_1\hat{a}_2 - \hat{a}_1\hat{a}^\dagger_2)$, where $\hat{a}_1$ and $\hat{a}_2$ are the annihilation operators for modes of the system $1$ and $2$, respectively.
Let a logical qubit be encoded in the system $1$ as $\hat{\rho}_{\Delta}$ by $\ket{\mu_\Delta}$, and the state in system $2$ be engineered as $\ket{0_\Delta}$.
Denoting the transmissivity from system $1$ to $2$ as $\eta := \sin^2 gt$, the output of the quantum channel $\EE_{\eta}$ for the logical qubit state $\hat{\rho}$ is given by tracing out the first subsystem of $ e^{i\hat{H}t}\hat{\rho}_{\Delta}\otimes \ket{0_\Delta}\!\bra{0_\Delta} e^{-i\hat{H}t}$.
It is well-known that if system $2$ is initially a Gaussian state, transmitting quantum information via $\EE_{\eta}$ is impossible when $\eta \le 1/2$ unless a non-Gaussian resource is used in system 2~\cite{Lami2020PRL}. Remarkably, in the highly non-Gaussian case where $\Delta \to 0$, we have the perfect recovery if $gt = \arctan \frac{2m_1+1}{2m_2}$ $(m_i\in\mathbb{Z}, m_2\ne 0)$, which implies the possibility of high fidelity quantum transduction at an arbitrarily small $\eta$~\cite{wang2024transduction, SuppMat}.

In the physical case where the GKP encoding has finite energy, we do not observe the perfect quantum transduction at $\eta \to 0$. The infidelity $1-F^{\mathrm{op}}_e$ is usually far from $0$ as displayed in Fig.~\ref{fig:curves_gkp}(a). However, Fig.~\ref{fig:curves_gkp}(b) shows that the gap between $F^{\mathrm{op}}_e$ and $F^{\mathrm{TC}}_e$ remains small and even negligible compared to $1-F^{\mathrm{op}}_e$ for some special values of $\eta$. As a comparison, Fig.~\ref{fig:curves_gkp}(c) shows that the norm of the commutator in Eq.~\eqref{eq:commutator_iten} (normalized by the factor $\mathcal{N}:=\|M\|_F\|I_\LL\otimes\tr_\LL\sqrt{M}\|_F$) approaches zero at these special values of $\eta$, such as $\frac{1}{5}$, $\frac{1}{2}$, $\frac{9}{13}$, and $\frac{4}{5}$. This phenomenon can be understood from the marginal distribution of the final wave function. The highly squeezed peaks from the GKP encoding are shuffled into orthogonal subspaces by $e^{-i\hat{H}t}$, thereby yielding an $M$ that approximates Corollary~\ref{cor:M_sigma} for some special values of $\eta$. More details regarding the simulation are provided in~\cite{SuppMat}.

\textit{Discussion.} ---
This work employs the tool of the QEC matrix to investigate entanglement fidelity under the Petz map, proposing a convenient expression for the entanglement fidelity in terms of the QEC matrix. We establish criteria for the Petz map to serve as the optimal recovery channel that maximizes entanglement fidelity. In this context, it generalizes the previously known KL conditions to all cases with this property. It sheds light on various applications, such as approximate quantum error correction and the performance of recovery maps. Although it is not yet known whether an alternative equivalent physical interpretation of $B\succeq 0$ exists, it serves as a useful tool in quantum information theory for exploring channels with special structures. We anticipate the physical implementation of such optimal Petz maps in future work.

We thank Ming Yuan for the insightful discussions. We acknowledge support from the ARO (W911NF-23-1-0077), ARO MURI (W911NF-21-1-0325), AFOSR MURI (FA9550-19-1-0399, FA9550-21-1-0209, FA9550-23-1-0338), DARPA (HR0011-24-9-0359, HR0011-24-9-0361), NSF (OMA-1936118, ERC-1941583, OMA-2137642, OSI-2326767, CCF-2312755), NTT Research, Packard Foundation (2020-71479), and the Marshall and Arlene Bennett Family Research Program. This material is based upon work supported by the U.S. Department of Energy, Office of Science, National Quantum Information Science Research Centers and Advanced Scientific Computing Research (ASCR) program under contract number DE-AC02-06CH11357 as part of the InterQnet quantum networking project. 

%

\clearpage
\widetext
\begin{center}
\textbf{\large Supplemental Material for ``Optimality Condition for the Petz Map''}
\end{center}
\setcounter{equation}{0}
\setcounter{figure}{0}
\setcounter{table}{0}
\setcounter{theorem}{0}
\makeatletter
\setcounter{secnumdepth}{3}
\renewcommand{\theequation}{S\arabic{equation}}
\renewcommand{\thefigure}{S\arabic{figure}}
\newcommand*\circled[1]{\tikz[baseline=(char.base)]{\node[shape=circle,draw,inner sep=1.0pt] (char) {#1};}}

\tableofcontents

\section{Penrose Graphical Notation}
In this Supplemental Material, we find it convenient to represent operators acting on multiple Hilbert spaces using Penrose tensor diagrams~\cite{wood2015tensornetworksgraphicalcalculus}. Most operators of interest involve the Hilbert spaces $\LL:=\mathbb{C}^d$, $\HH:=\mathbb{C}^n$, and $\KK:=\mathbb{C}^{n_K}$, which are indicated by blue, red, and black lines, respectively. For example, an operator $\sigma\in\Pos(\LL)$ and its canonical purification $\ket{\psi_\sigma} \in \LL\otimes \LL$ are represented as follows:
\begin{equation}
    \sigma = \vcenter{\hbox{\includegraphics[width=1.0cm]{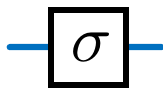}}}\;,\quad
    \ket{\psi_\sigma} = \vcenter{\hbox{\includegraphics[width=1.05cm]{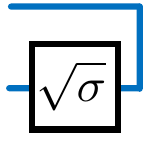}}}
    \;.
\end{equation}
The Choi states associated with the noise channel $\EE:\Pos(\LL)\to \Pos(\HH)$ and the recovery map $\RR:\Pos(\HH)\to \Pos(\LL)$ can be expressed via the operator-sum representation as follows:
\begin{equation}
    C_{\EE} = \vcenter{\hbox{\includegraphics[width=2.5cm]{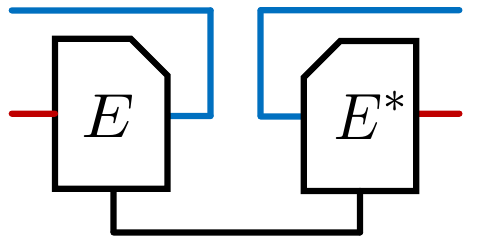}}}\;,\quad
    C_{\RR} = \vcenter{\hbox{\includegraphics[width=2.5cm]{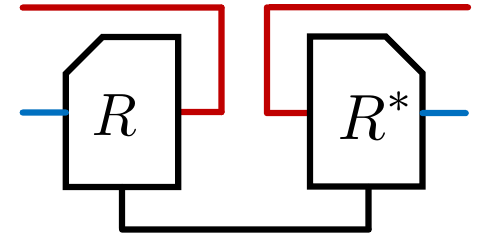}}}\;,\quad
    \hat{E}_k = \raisebox{-0.2cm}{$\vcenter{\hbox{\includegraphics[width=1.1cm]{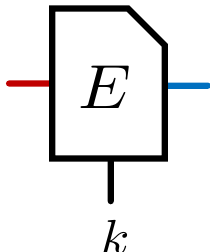}}}$}\;,\quad
    \hat{R}_k = \raisebox{-0.2cm}{$\vcenter{\hbox{\includegraphics[width=1.1cm]{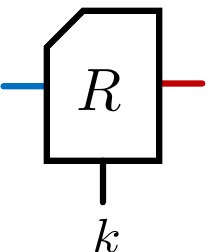}}}$}\;,
\end{equation}
where ``$E$'' and ``$R$'' are Kraus operators for $\EE$ and $\RR$. The symbol ``$*$'' denotes the complex conjugate of tensor components. The label ``$k$'' placed on a specific leg indicates that the corresponding index is fixed to the value $k$. Connections between legs represent tensor contractions over the corresponding indices.
Consequently, the entanglement fidelity can be expressed as:
\begin{align}
    F_e(\rho,\RR\circ\EE)&=\bra{\psi_\rho}\II\otimes \left(\RR\circ\EE\right)(\psi_\rho)\ket{\psi_\rho}\\
    &=\vcenter{\hbox{\includegraphics[width=5.0cm]{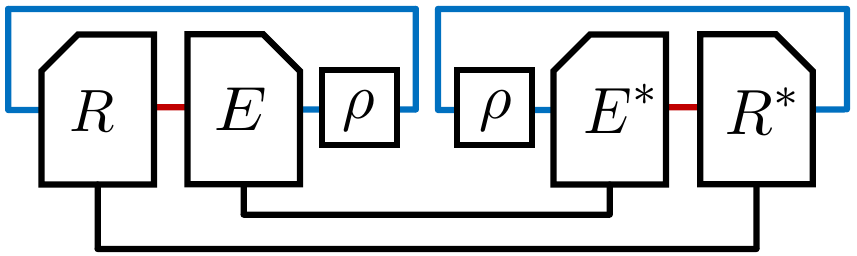}}}\;.\label{eq:ent_fid_penrose}
\end{align}
Similarly, the QEC matrices $M_\sigma$ and $M$ are represented by
\begin{equation}\label{eq:M_penrose}
    M_{\sigma} := \vcenter{\hbox{\includegraphics[width=3.5cm]{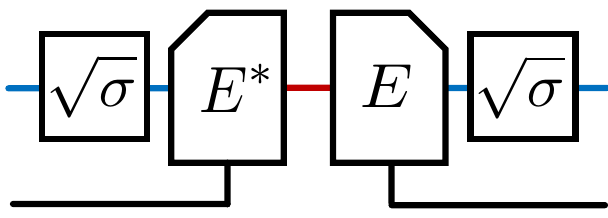}}}\;,\quad
    M = M_{I_\LL} = \vcenter{\hbox{\includegraphics[width=2.25cm]{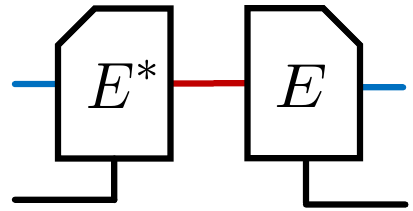}}}\;.
\end{equation}

By comparing the graphical representations of $C_{\EE}$ and $M$, we observe that the former traces out the degrees of freedom in $\KK$, whereas the latter traces out those in $\HH$. Thus, $M^T$ corresponds to the Choi state $C_{\EE^c}$, where $\EE^c$ denotes the complementary channel of $\EE$.

Additionally, the Kraus operator $\hat{E}_k:\LL\to\HH$ considered here acts exclusively on $\LL$. This implies that $\hat{E}_k = \hat{E}_k\Proj_\LL$, where $\Proj_\LL \equiv I_\LL$ is the projector onto the logical space $\LL$. Consequently, the KL conditions $\Proj_\LL \hat{E}_k^\dagger\hat{E}_\ell\Proj_\LL = \alpha_{k\ell}\Proj_\LL$ can be rewritten succinctly as $M = I_\LL\otimes \alpha$.

\section{The Details of Proofs}
\subsection{The Formula of Entanglement Fidelity}
\begin{theorem} \label{thm:ent_fid_petz_S}
Let $\rho, \sigma\in\Pos(\LL)$. The entanglement fidelity associated with the input state $\rho$ and the quantum channel $\Phi=\RPetz\circ\EE$ has a compact form given by
    \begin{equation}\label{eq:ent_fid_petz_S}
        F_e(\rho,\Phi) = \left\|\tr_{\LL}\left(M_\sigma^{-\frac{1}{2}}\left(\sqrt{\sigma}\otimes I_{\KK} \right)M\left(\rho\otimes I_{\KK}\right)\right)\right\|^2_F\;.
    \end{equation}
\end{theorem}
\begin{proof}
    The key here is show that $\EE(\sigma)^{q}$ is reasonable for all $q\in\mathbb{R}$:
    \begin{equation}
        \EE(\sigma)^{q} = \vcenter{\hbox{\includegraphics[width=4.5cm]{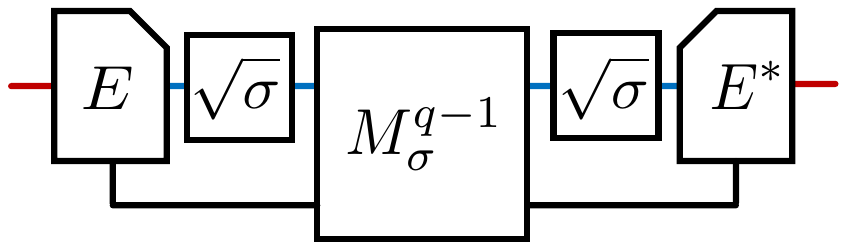}}}\;.
    \end{equation}
    One can verify that, $\EE(\sigma)^{-q} \EE(\sigma)^{q}$ is indeed the projector on $\supp(\EE(\sigma))$.
    By definition, the Kraus operator of the Petz map is given by
    \begin{equation}\label{eq:R_petz_penrose}
        \hat{R}^{\mathrm{P}}_{\sigma,\EE,k} := \sigma^{\frac{1}{2}}\hat{E}^\dagger_k\EE(\sigma)^{-\frac{1}{2}}
        =\vcenter{\hbox{\includegraphics[width=3.25cm]{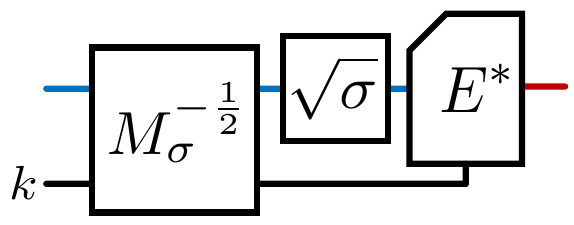}}}\;.
    \end{equation}
    Therefore, the entanglement fidelity according to Eq.~\eqref{eq:ent_fid_penrose} is
    \begin{align}
        F_e(\rho,\RPetz\circ\EE) &= \bigg\|\vcenter{\hbox{\includegraphics[width=4.5cm]{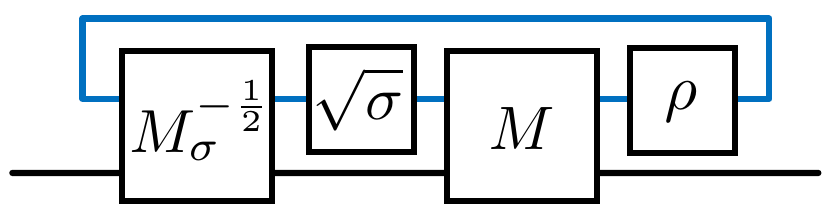}}}\bigg\|_F^2 = \left\|\tr_{\LL}\left(M_\sigma^{-\frac{1}{2}}\left(\sqrt{\sigma}\otimes I_{\KK} \right)M\left(\rho\otimes I_{\KK}\right)\right)\right\|^2_F\;.
    \end{align}
\end{proof}
\subsection{The Variation of the Completely Positive Map}
This section introduces the variation of Kraus operators in the operator-sum representation. We construct a CP map $\RR(\varepsilon)\sim\{\hat{R}_k\}$ with the following Kraus operator:
\begin{equation}\label{eq:Rvar_Kraus}
    \hat{R}_k(\varepsilon) := \big(\RP_{\sigma,\EE}\big)_k + \sum_{s=1}^{\infty}\varepsilon^s \hat{V}^{(s)}_k\;,
\end{equation}
such that $\RR(0) = \RPetz$.
For simplicity, we only consider those $\hat{V}^{(s)}_k:\supp(\EE(\sigma))\to \LL$ so that it guarantees $\hat{V}^{(s)}_\ell = \hat{V}^{(s)}_\ell\sum_k\hat{R}^\dagger_k\hat{R}_k|_{\epsilon = 0}$.
We assume that $\hat{V}^{(1)}_k$ satisfies
\begin{equation}\label{eq:RVVR_S}
\ddf{}{\varepsilon}\sum_k\hat{R}^\dagger_k(\varepsilon)\hat{R}_k(\varepsilon)\big|_{\varepsilon = 0} = 
\sum_k \hat{R}^{\dagger}_k(0)\hat{V}^{(1)}_k + \hat{V}^{(1)\dagger}_k \hat{R}_k(0) = 0\;,
\end{equation}
and $V^{(s)}_k$ $(s\ge 2)$ is defined recursively by:
\begin{align}\label{eq:V_W_def}
    \hat{V}^{(s)}_k &= -\frac{1}{2}\hat{R}_k(0) \hat{W}^{(s)} \;,\quad
    \hat{W}^{(s)} = \sum_k\sum_{r=1}^{s-1} \hat{V}^{(r)\dagger}_k \hat{V}^{(s-r)}_k\;.
\end{align}
In this way, Eq.~\eqref{eq:RVVR_S} and Eq.~\eqref{eq:V_W_def} together yield
\begin{align}
    \sum_k\hat{R}^\dagger_k(\varepsilon)\hat{R}_k(\varepsilon)
    &=
    \sum_k\hat{R}^\dagger_k(0)\hat{R}_k(0) + \sum_{s = 2}^{\infty}\varepsilon^s \left[\hat{W}^{(s)} + \sum_k \left(\hat{R}^\dagger_k(0) \hat{V}^{(s)}_k + \hat{V}^{(s)\dagger}_k \hat{R}_k(0) \right)\right]\\
    &=
    \sum_k\hat{R}^\dagger_k(0)\hat{R}_k(0) + \sum_{s = 2}^{\infty}\varepsilon^s \left(\hat{W}^{(s)}-\frac{1}{2}\hat{W}^{(s)} - \frac{1}{2}\hat{W}^{(s)}\right)
    =
    \sum_k\hat{R}^\dagger_k(0)\hat{R}_k(0)\;.
\end{align}
The construction described above ensures that $\RR(\varepsilon)$ is TP within $\supp(\EE(\sigma))$ for all $|\varepsilon|< v^{-1}$, with some $v>0$. Since $\|\hat{V}^{(s)}\|_1$ at most grows exponentially $\mathcal{O}(v^s)$ with some $v>0$. 
It follows that the operator $\hat{R}_k(\varepsilon)$ is also bounded. Since for all $\ket{\phi}\in\HH$, we have $\bra{\phi} I_\HH \ket{\phi} \ge \max_k \bra{\phi} \hat{R}^\dagger_k(\varepsilon) \hat{R}_k(\varepsilon) \ket{\phi}$, which leads to $1\ge \max_k \|\hat{R}_k(\varepsilon)\|_2$. Here, we have used the definition of the operator $p$-norm: $\| A\|_p :=\sup_{x\ne 0} \|Ax\|_p/\|x\|_p$.

In order to ensure the equal sign in Eq.~\eqref{eq:RVVR_S}, 
let $V^{(1)}_k$ be defined as follows, using a matrix $A$:
\begin{equation}\label{eq:V1_def}
 \hat{V}^{(1)}_k = \vcenter{\hbox{\includegraphics[width=2.5cm]{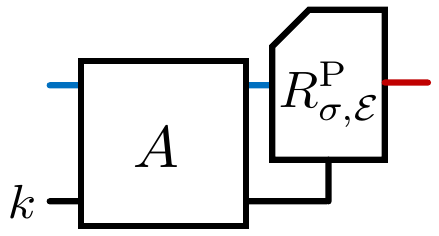}}}\;.
\end{equation}
By substituting Eq.~\eqref{eq:V1_def} into Eq.~\eqref{eq:RVVR_S} and then sandwiching the result with $\RP_{\sigma,\EE,k}(\cdots)\hat{R}^{\mathrm{P}\dagger}_{\sigma,\EE,\ell}$, we obtain 
\begin{equation}\label{eq:PAAP}
    \Proj_{M_\sigma}(A+A^\dagger) \Proj_{M_\sigma} = 0\;.
\end{equation}
In Eq.~\eqref{eq:PAAP}, the projector on $\supp(M_\sigma)$ is denoted as $\Proj_{M_\sigma}$, which comes from
\begin{equation}
    \vcenter{\hbox{\includegraphics[width=8.5cm]{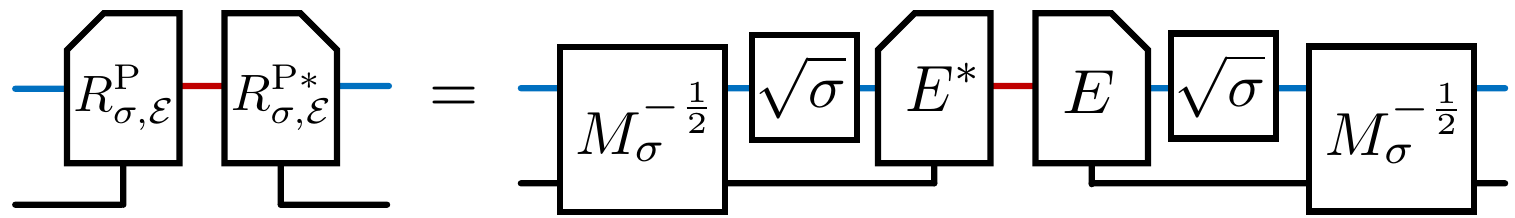}}}
    =M_\sigma^{-\frac{1}{2}}M_\sigma M_\sigma^{-\frac{1}{2}}
    =
    \Proj_{M_\sigma}\;,
\end{equation}
where Eq.~\eqref{eq:R_petz_penrose} and Eq.~\eqref{eq:M_penrose} are used for the derivation.
In other words, to ensure that the recovery map $\RR(\varepsilon)$ is TP on $\supp(\EE(\sigma))$, the matrix $A$ can be arbitrary, provided it is anti-Hermitian on $\supp(M_\sigma)$.

Let $\supp(\rho)\subseteq \supp(\sigma)$, we define
\begin{align}
    \gamma&:=\sigma^{-\frac{1}{2}}\rho\;,\\
    T&:=\tr_\LL \left(\left(\gamma^\dagger \otimes I_\KK\right) \sqrt{M_\sigma}\right) = \vcenter{\hbox{\includegraphics[width=4.25cm]{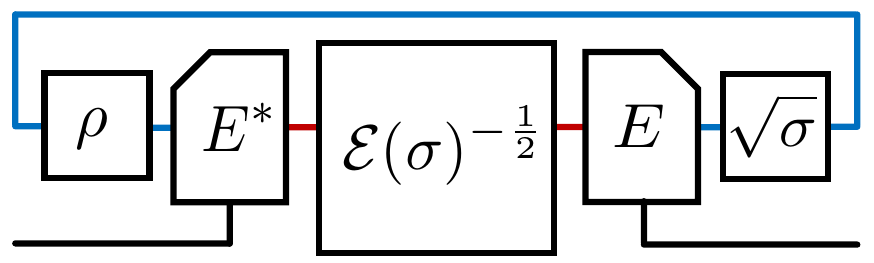}}}\;,\label{eq:T_def_S} \\
    B&:=\sqrt{M_\sigma}\left(\gamma\otimes T\right)\;.\label{eq:B_def_S}
\end{align}
The second equal sign in Eq.~\eqref{eq:T_def_S} is obtained due to 
\begin{equation}
    M_\sigma^{q} = \vcenter{\hbox{\includegraphics[width=4.25cm]{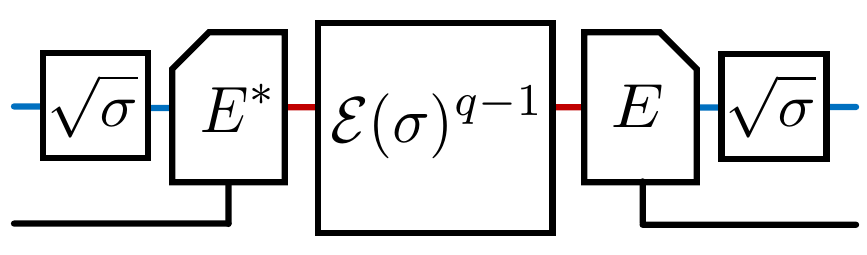}}}\;.
\end{equation}
The above ingredients allow us to simplify the derivative of $F_e$ at $\varepsilon = 0$ as:
\begin{align}
    \ddf{}{\varepsilon} F_e(\rho,\RR(\varepsilon)\circ\EE)\big|_{\varepsilon = 0} 
    &= 2\mathrm{Re} \sum_{k,\ell=1}^{n_K} 
    \tr\left(\hat{V}^{(1)}_k\hat{E}_\ell\rho\right)
    \tr\left(\rho\hat{E}^\dagger_\ell \hat{R}^{\dagger}_k(0)\right) \label{eq:dFde_sub1}\\
    &= 2\mathrm{Re}\vcenter{\hbox{\includegraphics[width=6cm]{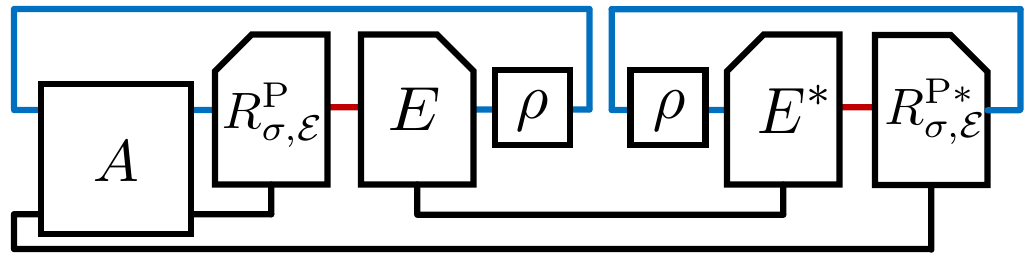}}} \\
    &=2\mathrm{Re}\;\tr \left(\vcenter{\hbox{\includegraphics[width=4.5cm]{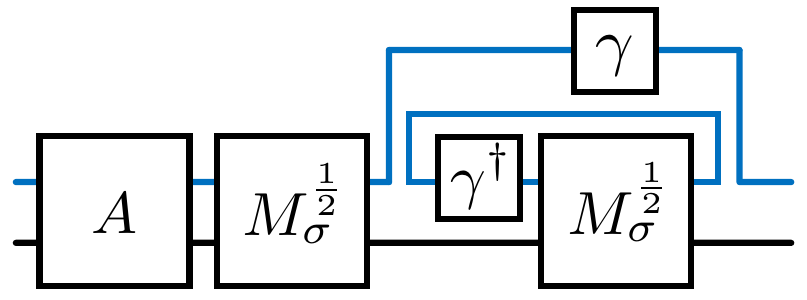}}}\right) \label{eq:dFde_sub2}\\
    &= 2\mathrm{Re}\;\tr \left(AB\right)
    = \tr\left(A\left(B-B^\dagger\right)\right)
    \;, \label{eq:dFde}
\end{align}
where we have used Eq.~\eqref{eq:M_penrose}, Eq.~\eqref{eq:R_petz_penrose}, and Eq.~\eqref{eq:V1_def} to obtain Eq.~\eqref{eq:dFde_sub2}. Particularly, we have used the anti-Hermiticity in Eq.~\eqref{eq:PAAP} for the last equal sign. 

Since the second derivative involves $\hat{V}^{(2)}_k$, we can combine Eq.~\eqref{eq:V_W_def} and Eq.~\eqref{eq:V1_def} to show that
\begin{equation}\label{eq:V2_def}
    \hat{V}_k^{(2)} = -\frac{1}{2}\hat{R}_k(0)\sum_{\ell = 1}^{n_K} \hat{V}_{\ell}^{(1)\dagger} \hat{V}_\ell^{(1)} = -\frac{1}{2}\vcenter{\hbox{\includegraphics[width=3.5cm]{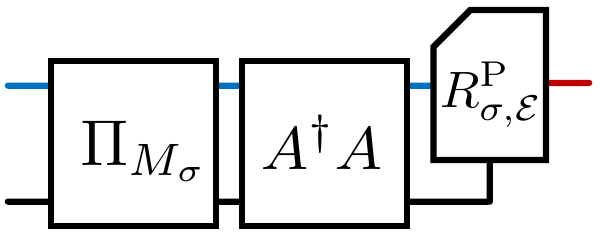}}}\;.
\end{equation}
Thus, the second derivative is:
\begin{align}
    \ddf{^2}{\varepsilon^2} F_e(\rho,\RR(\varepsilon)\circ\EE)\big|_{\varepsilon = 0} 
    &= 2\sum_{k,\ell=1}^{n_K}\left|\tr\left(\hat{V}_k^{(1)}\hat{E}_\ell \rho\right)\right|^2 + 4\mathrm{Re}\sum_{k,\ell=1}^{n_K} 
    \tr\left(\hat{V}_k^{(2)}\hat{E}_\ell\rho\right)
    \tr\left(\rho\hat{E}^\dagger_\ell \hat{R}^{\dagger}_k(0)\right) \label{eq:dFde2_sub0}\\
    &=
    2\sum_{k,\ell=1}^{n_K}\left|\tr(\hat{V}_k^{(1)}\hat{E}_\ell \rho)\right|^2 - 2\mathrm{Re}\;\tr\left(\Proj_{M_\sigma}A^\dagger AB\right) \label{eq:dFde2_sub1} \\
    &=
    2\sum_{k,\ell=1}^{n_K}\left|\tr(\hat{V}_k^{(1)}\hat{E}_\ell \rho)\right|^2 - \tr\left(A^\dagger A\left(B\Proj_{M_\sigma}+\Proj_{M_\sigma}B^\dagger \right)\right)\;, \label{eq:dFde2}
\end{align}
where the second term of Eq.~\eqref{eq:dFde2_sub1} is derived by substituting Eq.~\eqref{eq:V2_def} into Eq.~\eqref{eq:dFde2_sub0}. Eq.~\eqref{eq:dFde2} is then obtained from Eq.~\eqref{eq:dFde2_sub1} by permuting matrix products within $\tr(\cdots)$.

\subsection{The Optimality Condition}
\begin{theorem}\label{thm:condition_S}
    Let $\rho, \sigma\in\Pos(\LL)$ with $\supp(\rho) \subseteq \supp(\sigma)$, and let $M_\sigma$ be the QEC matrix associated with the reference state $\sigma$ and the noise channel $\EE$. 
    Then, $\RR = \RPetz$ optimizes  $F_e(\rho, \RR\circ\EE)$ if and only if 
    \begin{equation}
        B:=\sqrt{M_\sigma}\left(\gamma\otimes T\right)\succeq 0\;.    
    \end{equation}
\end{theorem}
\begin{proof}
    To prove the \textit{necessity}, we apply the variation of $\RR$, as defined by Eq.~\eqref{eq:Rvar_Kraus}. Since $\RR(\varepsilon)$ optimizes $F_e$ at $\varepsilon = 0$ with any $A$ fulfilling Eq.~\eqref{eq:PAAP}, we have the following conditions 
    \begin{align}
        \ddf{}{\varepsilon} F_e(\rho,\RR(\varepsilon)\circ \EE)|_{\varepsilon = 0} &= 0\;, \label{eq:dFde_cond}\\
        \ddf{^2}{\varepsilon^2} F_e(\rho,\RR(\varepsilon)\circ \EE)|_{\varepsilon = 0} &\le 0\;, \label{eq:dFde2_cond}
    \end{align} 
    Eq.~\eqref{eq:dFde_cond} requires that $B = B^\dagger$, implying that $\supp(B)\subseteq \supp(M_\sigma)$. Additionally, Eq.~\eqref{eq:dFde2_cond} demands that the second term of Eq.~\eqref{eq:dFde2} be non-negative, leading to the condition $B\Proj_{M_\sigma}+\Proj_{M_\sigma}B^\dagger = B + B^\dagger \succeq 0$. Therefore, $B$ must be positive semidefinite (PSD).

    To prove the \textit{sufficiency} of $B\succeq 0$, we consider the KKT conditions in semidefinite programming of the following convex optimization problem~\cite{watrous2009SDP}:
    \begin{equation}
        \begin{array}{lll}
            &\text{minimize}\quad -F_e[C_\RR]  \\
            &\text{subject to}\quad G:=I_{\HH} - \tr_{\LL} C_{\RR} \succeq 0
            \text{ and }C_\RR\succeq 0\;,\\
        \end{array}
    \end{equation}
    where $-F_e[C_\RR] = -\tr\left(C_\RR\nabla_{C_\RR}F_e(\rho,\RR\circ\EE)\right)$ is a differentiable convex function with respect to the Choi matrix $C_\RR$. The domain of $C_\RR$ is also convex, as any convex combination of Choi matrices defines a valid completely positive map.
    The definiteness of $G$ restricts $C_\RR$ to be trace nonincreasing.
    The optimization problem is equivalent to the Lagrange function:
    \begin{align}
        \mathscr{L}(C_\RR,\Lambda,\Gamma)&:= \mathscr{F}(C_\RR)  - \tr(\Lambda^TG) - \tr(\Gamma^T C_\RR)\\
        &=
        \tr\left(\left(-\nabla_{C_\RR}F_e + I_\LL\otimes \Lambda^T - \Gamma^T\right)C_{\RR}\right) - \tr(\Lambda)\;,
    \end{align}
    with Lagrange multipliers $\Lambda$ and $\Gamma$ that penalize the constraints $G \succeq 0$ and $C_\RR \succeq 0$, respectively.
    According to the Karush–Kuhn–Tucker (KKT) conditions, if $C_{\RR} = C_*$ is the optimal solution, then there exist solutions of $\Lambda$ and $\Gamma$ such that the following conditions hold at $C_\RR = C_*$.:
    \begin{itemize}
        \item {Primal feasibility:} $C_\RR\succeq 0$ and $G \succeq 0$.
        \item {Dual feasibility:} $\Gamma \succeq 0$, $\Lambda \succeq 0$.
        \item {Complementary slackness:} $-\tr(\Lambda^TG) - \tr(\Gamma^T C_\RR) = 0$. 
        \item {Stationarity:} $\nabla_{C_{\RR}} \mathscr{L}=0$.
    \end{itemize}
    Assuming $\RR_* = \RPetz$, the primal feasibility condition is satisfied. It remains to ensure that the other conditions hold under $B \succeq 0$. The strategy involves constructing $\Lambda$ and $\Gamma$ based on some of these conditions.
    We begin by imposing stationarity, which requires $0=-\nabla_{C_\RR}F_e + I_\LL\otimes \Lambda^T - \Gamma^T$. This leads to
    \begin{equation}\label{eq:Gamma_def}
        \Gamma = I_{\LL}\otimes \Lambda - \nabla_{C_\RR^T}F_e\;.
    \end{equation}
    Additionally, given that $\Gamma$, $\Lambda$, $G$, and $C_*$ are all PSD, the condition of complementary slackness is equivalent to
    \begin{align}\label{eq:GammaC_0}
        \Gamma^T C_* &= 0\;,
    \end{align}
    and
    \begin{align}\label{eq:LambdaG_0}
        \Lambda^TG &= \Lambda^T\left(I_{\HH} -  \tr_{\LL} C_*\right) = 0\;.
    \end{align}
    According to Eq.~\eqref{eq:GammaC_0}, left multiplying Eq.~\eqref{eq:Gamma_def} with $C_*^T$ yields
    \begin{align}\label{eq:stationarity}
         0 = \tr_{\LL}\left(C_*^T\Gamma\right) = -\tr_{\LL}\left(C_*^T\nabla_{C_{\RR}^T}F_e\right) + \left(\tr_{\LL}C_*^T\right)\Lambda = -\tr_{\LL}\left(C_*^T\nabla_{C_{\RR}^T}F_e\right) + \Lambda\;.
    \end{align}
    By substituting $ \nabla_{C^T_{\RR}}F_e = \big(\rho^T\otimes I_\HH\big)C_{\EE}\big(\rho^T\otimes I_\HH\big)$ and $C_* = C_{\RR^{P}_{\sigma,\EE}}$, Eq.~\eqref{eq:stationarity} provides the solution of $\Lambda$:
    \begin{equation}\label{eq:Lambda_def}
        \begin{aligned}
            \Lambda &= \tr_{\LL} \left(C^T_{\RR_*}\nabla_{C^T_{\RR}}F_e\right) = 
             \raisebox{-0.15cm}{$\vcenter{\hbox{\includegraphics[width=10.5cm]{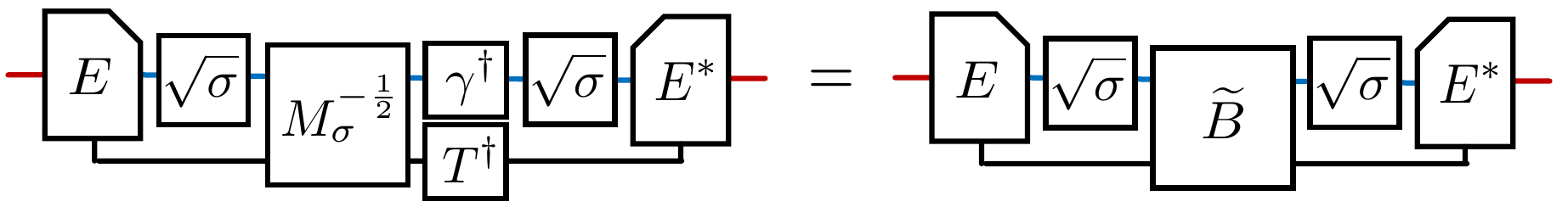}}}$}\;,
         \end{aligned}
    \end{equation}
    where we define
    \begin{equation}
        \widetilde{B} := M_\sigma^{-\frac{1}{2}} B^\dagger M_\sigma^{-\frac{1}{2}}\;.
    \end{equation}
    Given that $\RR_*=\RPetz$ is TP on $\supp(\EE(\sigma))$,
    the solution for $\Lambda$ implies that $\supp(\Lambda) \subseteq \supp(\EE(\sigma))\subseteq\ker\left(I_{\HH} -  \tr_{\LL} C_*\right)$. This ensures that Eq.~\eqref{eq:LambdaG_0} is satisfied. Therefore, $\Lambda$ and $\Gamma$ fulfill the condition of complementary slackness.

    For dual feasibility, since $\widetilde{B}\succeq 0$, Eq.~\eqref{eq:Lambda_def} must imply that $\Lambda\succeq 0$.
    Before proving the definiteness of $\Gamma$, we need to clarify the relationships between the matrix supports. 
    Let us define $\hat{S}_\sigma := \sum_k (I_\LL\otimes\hat{E}_k)\ket{\psi_\sigma}\bra{k}$.
    Then, from Eq.~\eqref{eq:T_def_S}, we have
    \begin{equation}
        TT^\dagger = \hat{S}^\dagger_{\rho^2} \left(I_\LL\otimes \EE(\sigma)^{-\frac{1}{2}}\right)\II\otimes\EE(\psi_\sigma) \left(I_\LL\otimes \EE(\sigma)^{-\frac{1}{2}}\right) \hat{S}_{\rho^2}\;.
    \end{equation}
    Assuming $\supp(\rho)\subseteq \supp(\sigma)$, we obtain
    \begin{equation}\label{eq:subseteq_seq}
        \supp(\hat{S}_{\rho}\hat{S}_{\rho}^\dagger) = \supp(\II\otimes \EE(\psi_{\rho})) \subseteq \supp(\II\otimes \EE(\psi_\sigma)) \subseteq \supp(I_\LL\otimes \EE(\sigma)) = 
        \supp(I_\LL\otimes \EE(\sigma)^{-\frac{1}{2}})\;.
    \end{equation}
    The second ``$\subseteq$'' in Eq.~\eqref{eq:subseteq_seq} is justified by the fact that the mixed state $I_\LL\otimes \EE(\sigma)$ can be viewed as a pure state $\II\otimes \EE(\psi_\sigma)$ that is maximally depolarized on the reference system  $\LL$. Eq.~\eqref{eq:subseteq_seq} implies
    \begin{equation}\label{eq:TTislarger}
         \supp(\hat{S}^\dagger_{\rho^2} \hat{S}_{\rho^2}) \subseteq \supp(TT^\dagger) \Longrightarrow \hat{S}_{\rho^2} TT^{-1} = \hat{S}_{\rho^2}\;.
    \end{equation}
    It follows that $\Gamma$ can be decomposed as:
    \begin{align}
        \Gamma &= \vcenter{\hbox{\includegraphics[width=8.0cm]{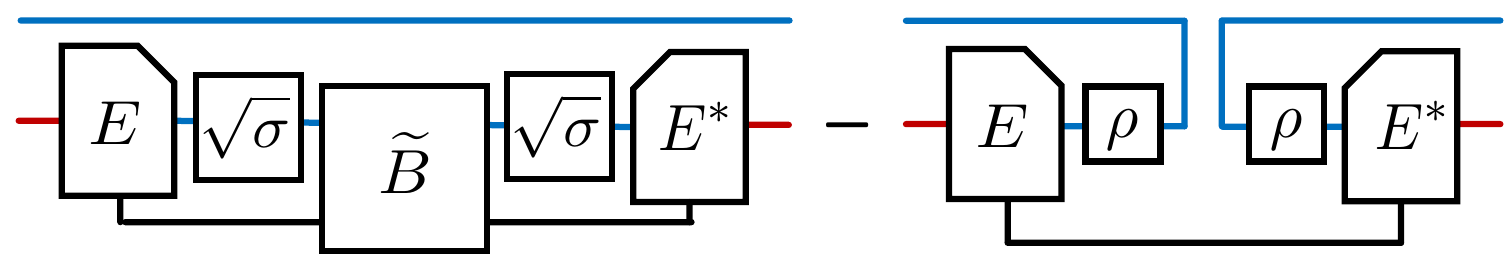}}}  \label{eq:Gamma_sub1}\\
        &= \raisebox{-0.07cm}{$\vcenter{\hbox{\includegraphics[width=11.0cm]{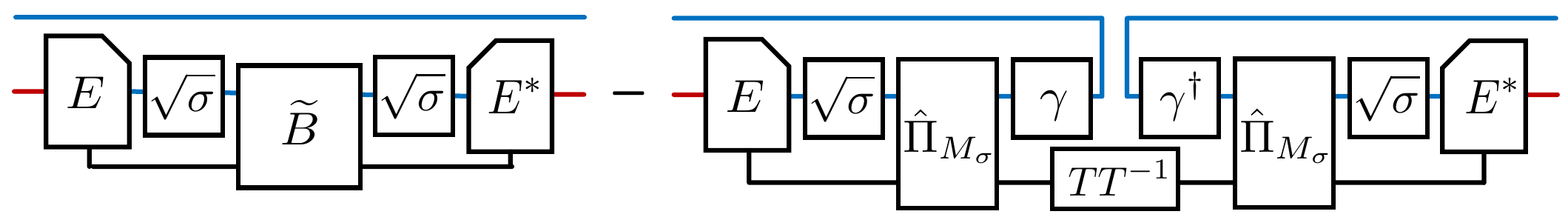}}} \label{eq:Gamma_sub2}$}\\
        &= \raisebox{-0.2cm}{$\vcenter{\hbox{\includegraphics[width=10.5cm]{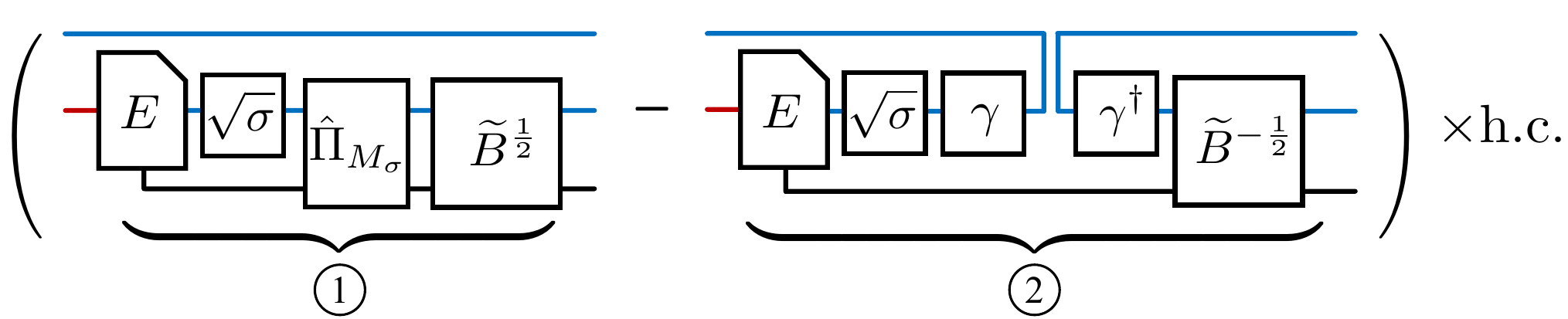}}} \label{eq:Gamma_sub3}$}\;,
    \end{align}
    where Eq.~\eqref{eq:Gamma_sub1} is derived from Eq.~\eqref{eq:ent_fid_penrose}, Eq.~\eqref{eq:Gamma_def}, and Eq.~\eqref{eq:Lambda_def}. 
    In Eq.~\eqref{eq:Gamma_sub2}, the projector $TT^{-1}\equiv T(T^\dagger T)^{-1}T^\dagger$ is inserted into the expression, justified by Eq.~\eqref{eq:TTislarger}. 
    Additionally, the insertion of the projector $\Proj_{M_\sigma}$ in the above derivation is self-evident due to the definition of $M_\sigma$.
    Finally, we explain how Eq.~\eqref{eq:Gamma_sub2} can be factorized into Eq.~\eqref{eq:Gamma_sub3} as follows. 
    It is straightforward to see that $\circled{1}\;\circled{1}^\dagger$ equals to the first term of Eq.~\eqref{eq:Gamma_sub2}. 
    Since $\widetilde{B}$ is PSD, we find
     \begin{align}
        \tr_\LL \left(\left(\gamma^\dagger \otimes I_\KK\right) \widetilde{B}^{-\frac{1}{2}}\widetilde{B}^{-\frac{1}{2}\dagger} \left(\gamma \otimes I_\KK\right)\right) &= \tr_\LL \left(\left(\gamma^\dagger \otimes I_\KK\right) \widetilde{B}^{\dagger-1} \left(\gamma \otimes I_\KK\right)\right)\\
        &=\tr_\LL \left(\left(\gamma^\dagger \otimes I_\KK\right) \sqrt{M_\sigma}\left(\gamma^{-1}\otimes T^{-1}\right)  \left(\gamma \otimes I_\KK\right)\right) \\
        &=\tr_\LL\left(\left(\gamma^\dagger\otimes I_\KK\right)\sqrt{M_\sigma}\right)T^{-1} = TT^{-1}\;, 
    \end{align}
    so there is
    \begin{align}
        \circled{2}\;\circled{2}^\dagger &= \text{the second term of Eq.~\eqref{eq:Gamma_sub2}}\;.
    \end{align}
    Utilizing $B\Proj_{M_\sigma} = \Proj_{M_\sigma}B = B$, we have
    \begin{align}
        \Proj_{M_\sigma} \widetilde{B}^{\frac{1}{2}}\widetilde{B}^{-\frac{1}{2}\dagger}  
        &= \Proj_{M_\sigma} \widetilde{B}^\dagger \widetilde{B}^{\dagger-1} = M_\sigma^{-\frac{1}{2}}M_\sigma^{\frac{1}{2}}(\gamma\otimes T)\Proj_{M_\sigma} (\gamma^{-1}\otimes T^{-1}) \\
        &= M_\sigma^{-\frac{1}{2}} B\Proj_{M_\sigma} (\gamma^{-1}\otimes T^{-1}) = M_\sigma^{-\frac{1}{2}} B (\gamma^{-1}\otimes T^{-1}) \\
        &= \Proj_{M_\sigma}(\gamma\otimes T) (\gamma^{-1}\otimes T^{-1})\\
        & = \Proj_{M_\sigma} \left(\gamma \gamma^{-1}\otimes TT^{-1}\right)\;,
    \end{align}
    so there is
    \begin{align}
        \circled{1}\;\circled{2}^\dagger = \circled{2}\;\circled{1}^\dagger = \circled{2}\;\circled{2}^\dagger\;. 
    \end{align}
    It follows that 
    \begin{equation}
        \begin{aligned}
            \Gamma &= \circled{1}\;\circled{1}^\dagger - \circled{2}\;\circled{2}^\dagger \\
            &= \circled{1}\;\circled{1}^\dagger - \circled{1}\;\circled{2}^\dagger - \circled{2}\;\circled{1}^\dagger + \circled{2}\;\circled{2}^\dagger \\
            &= \left(\circled{1} - \circled{2}\right)\left(\circled{1} - \circled{2}\right)^\dagger \succeq 0\;.
        \end{aligned}
    \end{equation}
    Hence, the explicit Lagrange multipliers $\Lambda$ and $\Gamma$ satisfy all the KKT conditions, thereby proving the sufficiency of the theorem.
\end{proof}
\subsection{Corollaries and Their Proofs}
\begin{corollary}\label{cor:rhosigma_comm_S}
    Let $\rho, \sigma\in\Pos(\LL)$ satisfying $[\rho,\sigma]=0$ with $\supp(\rho) \subseteq \supp(\sigma)$.
    Then, $\RR = \RPetz$ optimizes  $F_e(\rho, \RR\circ\EE)$ if and only if
    \begin{equation}\label{eq:commutator_eta_S}
        \left[M_\sigma, \gamma\otimes T\right] = 0.
    \end{equation}
\end{corollary}
\begin{proof}
    If $[\rho,\sigma] = 0$, then the operator $\gamma\otimes T=\gamma\otimes \tr_\LL\left(\left(\gamma^\dagger\otimes I_\KK\right)\sqrt{M_\sigma}\right)$ is Hermitian. Given that the optimality condition in Theorem~\ref{thm:condition_S} is $B\succeq 0$, this condition is equivalent to $\left[M_\sigma, \gamma\otimes T\right] = 0$ in this special case.
\end{proof}
\begin{corollary}\label{cor:M_sigma_S}
    Let $M_\sigma$ be the QEC matrix associated with the reference state $\sigma$ and the noise channel $\EE$, and $\KK = \bigoplus_s\KK_s$. If $\sqrt{M_\sigma} = \bigoplus_s \beta_s$, where $\beta_s\in\Pos(\LL\otimes \KK_s)$ satisfying $\tr_\LL \beta_s \propto I_{\KK_s}$, then $\RR = \RPetz$ optimizes $F_e(\sqrt{\sigma},\RR\circ\EE)$.
\end{corollary}
\begin{proof}
    When the input state of $F_e$ is $\rho = \sqrt{\sigma}$. The optimality condition (Eq.~\eqref{eq:commutator_eta_S}) reduces to $\left[M_\sigma,I_\LL\otimes \tr_\LL\sqrt{M_\sigma}\right]=0$. Meanwhile, since $\tr_\LL\sqrt{M_\sigma} = \bigoplus_s \left(\tr_\LL \beta_s\right) I_{\KK_s}$ must commute with $\sqrt{M_\sigma} = \bigoplus_s \beta_s$, it follows that the Petz map $\RR = \RPetz$ must attain optimality with respect to the objective function $F_e(\sqrt{\sigma}, \RR\circ\EE)$.
\end{proof}

\section{Optimality of the Petz Map: Analytical Examples}
\subsection{The Example with $B\succeq 0$ and $[\rho,\sigma]\ne 0$}
In this section, we construct a model in which the triplet $(\rho, \sigma, \EE)$ satisfies $B \succeq 0$ but does not satisfy $[\rho, \sigma] = 0$. According to Theorem~\ref{thm:condition_S}, the Petz map $\RR = \RPetz$ achieves optimality with respect to the objective function $F_e(\rho, \RR \circ \EE)$, despite failing to satisfy the generalized commutation relation in Eq.~\eqref{eq:commutator_eta_S}.

We consider both $\LL$ and $\KK$ to be partitioned into orthogonal subspaces:
\begin{align}
    \LL &= \bigoplus_{s} \LL_s\;,\quad
    \KK = \bigoplus_{s} \KK_s\;.
\end{align}
Let $\rho$ and $\sigma$ be density matrices in $\Pos(\LL)$ that satisfy $\tr \rho = \tr \sigma = 1$, without any commutation relation imposed.
We assume that the input state $\rho$ and the reference state $\sigma$ can be decomposed as direct sums: 
\begin{align}
    \rho &= \sum_{s} \rho_s,\quad
    \sigma = \sum_{s} \sigma_s,\quad \supp(\rho_s)\subseteq \supp(\sigma_s)\subseteq \LL_s \;.
\end{align}
Consequently, $\gamma$ also possesses a direct sum structure:
\begin{equation}
    \gamma = \sigma^{-\frac{1}{2}}\rho = \sum_{s} \sigma^{-\frac{1}{2}}_s\rho_s\;.
\end{equation}
Since $[\rho, \sigma] = 0$ is not required, $\gamma$ may be non-Hermitian.
We consider the noise channel $\EE$ with the following QEC matrix:
\begin{align}\label{eq:E_SLcond_directsum}
    M = \sum_{s} I_{\LL_s}\otimes \alpha_s,\quad \supp(\alpha_s)\subseteq \KK_s\;,
\end{align}
where $\alpha_s\succeq 0$, and $\tr\alpha_s = 1$ is required to ensure that $\tr_\KK M = I_\LL$. Since 
\begin{equation}
    \sqrt{M_\sigma} = \sum_{s} \sqrt{\sigma_s}\otimes \sqrt{\alpha_s}\;,
\end{equation}
we obtain $T\equiv\tr_\LL\left(\left(\gamma^\dagger\otimes I_\KK\right) \sqrt{M_\sigma}\right) = \sum_s (\tr\rho_s)\sqrt{\alpha_s}$, so
\begin{align}
    B &= \sqrt{M_\sigma}\left(\gamma\otimes T\right) 
    = \sum_{s} \left(\tr\rho_s\right)\rho_s \otimes \alpha_s \succeq 0\;.
\end{align}
Moreover, in this case, the entanglement fidelity $F_e(\rho,  \RPetz\circ \EE)$ (Eq.~\eqref{eq:ent_fid_petz_S}) is given by 
\begin{align}
F_e(\rho,  \RPetz\circ \EE) &= \left\|T\right\|_F^2 = \sum_{s}(\tr\rho_s)^2 \le 1\;.
\end{align}
Notably, when $\LL \equiv \LL_1$ and $\KK \equiv \KK_1$, the KL conditions are satisfied, immediately yielding the familiar result with $M = I_\LL \otimes \alpha$, $B = \rho \otimes \alpha$, and $F_e = 1$.

\subsection{The Example with $B\succeq 0$ and $[\rho,\sigma] = 0$, but $\rho\ne \sqrt{\sigma}$}
This subsection presents a simple example with $\LL = \mathbb{C}^2$ and $\KK = \mathbb{C}^2$, in which Eq.~\eqref{eq:commutator_eta_S} is satisfied but the optimality condition in~\cite{Iten2017TIT}:
\begin{equation}\label{eq:commutator_iten_S}
    \left[M_\sigma,I_\LL\otimes \tr_\LL\sqrt{M_\sigma}\right]=0
\end{equation}
is violated. Such an example is analytically constructed by the following parameterization of $\sigma$, $\gamma$, and $M_\sigma$:
\begin{equation}
    \sqrt{\sigma} = \begin{bmatrix}
        a&\\
        &b
    \end{bmatrix},\quad
    \gamma = \begin{bmatrix}
        s&\\
        &t
    \end{bmatrix},\quad
    \sqrt{M_\sigma} = \begin{bmatrix}
    u&&&\\
    &z&y&\\
    &y&x&\\
    &&&v\\
    \end{bmatrix}\;,
\end{equation}
where the following conditions are required due to $\rho, \sigma, \sqrt{M_\sigma} \succeq 0$ and $\tr\sigma = \tr\rho = 1$:
\begin{equation}\label{eq:constr_S}
    a,b,s,t>0,\quad a^2+b^2=1,\quad as +bt = 1\quad xz - y^2>0\;.
\end{equation}
It follows that $M$ can be determined as:
\begin{equation}
    M = \left(\sigma^{-\frac{1}{2}}\otimes I_\KK\right) M_\sigma \left(\sigma^{-\frac{1}{2}}\otimes I_\KK\right)= \begin{bmatrix}
    \frac{u^2}{a^2}&&&\\
    &\frac{y^2+z^2}{a^2}&\frac{(x+z)y}{ab}&\\
    &\frac{(x+z)y}{ab}&\frac{x^2+y^2}{b^2}&\\
    &&&\frac{v^2}{b^2}\\
    \end{bmatrix}\;,
\end{equation}
By the definition of $M$, the Kraus operators ($\hat{E}_i$) of $\EE$ can be derived from the submatrices of $\sqrt{M}$.
Since $M$ acts on $\LL\otimes \KK$, the TP condition $\tr_{\LL}M = I_{\KK}$ is equivalent to
\begin{equation}\label{eq:hypersurf}
\left\{
    \begin{aligned}
    u^2+y^2+z^2 - a^2 &= 0 \\
    v^2+x^2+y^2 - b^2 &= 0 
    \end{aligned}\right.\;. 
\end{equation}
Next, we need to ensure that the commutator in Eq.~\eqref{eq:commutator_eta_S} vanishes. Note that $\sqrt{M_\sigma}$ is not diagonal, while
\begin{align}
    \gamma\otimes T = \gamma\otimes \tr_{\LL}\left(\left(\gamma^\dagger\otimes I_{\KK}\right)\sqrt{M_\sigma}\right)
    &=
    \begin{bmatrix}
        s&\\
        &t
    \end{bmatrix}\otimes
    \begin{bmatrix}
        us + xt&\\
        &zs + vt
    \end{bmatrix}
    \\&= 
    \begin{bmatrix}
    (us + xt)s&&&\\
    &(zs + vt)s&&\\
    &&(us + xt)t&\\
    &&&(zs + vt)t\\
    \end{bmatrix}\;.
\end{align}
Thus, the condition
\begin{equation}\label{eq:hypeplane_1}
    (zs + vt)s = (us + xt)t
\end{equation}
is needed for the commutation relation in Eq.~\eqref{eq:commutator_eta_S} to hold.
By combining Eq.~\eqref{eq:hypersurf}, Eq.~\eqref{eq:hypeplane_1} and the constraints in Eq.~\eqref{eq:constr_S}, we can explicitly determine the undetermined parameters as follows:
\begin{align}
    b &= \sqrt{1-a^2},\quad s = \frac{1-bt}{a},\quad v = \sqrt{b^2 - x^2 - y^2},\quad z = \alpha u + \beta,\quad p_2 u^2 + p_1 u + p_0 = 0\;.
\end{align}
where $p_2 := 1+\alpha^2$, $p_1 := 2\alpha\beta$, $p_0 := \beta^2 - a^2 + y^2$, $\alpha := t/s$, and  $\beta := x \alpha^2 - v \alpha$.
The values of $a$, $x$, $y$, and $t$ can be adjusted within a wide range to yield many feasible solutions. For instance, if $a=\sqrt{0.3}$, $t = 1$, $x = 0.25$ and $y = 0.08$, then we have $b = \sqrt{0.7}$, $s \approx 0.298217$, $z \approx 0.529706$, $u \approx 0.114068$ and $v \approx 0.794418$. It can be verified that
\begin{equation}
    \max_{\RR} F_{e}(\rho,\RR\circ\EE) = F_{e}(\rho,\RRP_{\sigma,\EE} \circ\EE) = \left\|T\right\|_F^2\approx 0.987703\;.
\end{equation} 
In this case, we have ensured that $\rho \ne \sigma$, Eq.~\eqref{eq:commutator_iten_S} is violated, and Eq.~\eqref{eq:commutator_eta_S} is satisfied. 

\subsection{Examples with $\rho = \sigma = I_\LL/d$}
In this section, we present examples in which the KL conditions are violated, yet the Petz map still optimizes the channel fidelity for the case $\rho = \sigma = I_\LL/d$, where $d \equiv \dim \LL$. In this special case, the Petz map is referred to as the transpose channel (TC), and the entanglement fidelity is sometimes called the channel fidelity.
For simplicity, we denote $\RR^{\mathrm{TC}}:=\RRP_{I_\LL,\EE}$, $\hat{R}^{\mathrm{TC}}_k:=(\RP_{I_\LL,\EE})_k$, and $F^{\mathrm{TC}}_e = F_e(I_\LL/d,\RR^{\mathrm{TC}}\circ\EE)$. The optimal channel fidelity is denoted as $F^{\mathrm{op}}_e = \max_\RR F_e(I_\LL/d, \RR\circ\EE)$.

\subsubsection{A Toy Model}
We present a simple example where the TC is always optimal. Consider the following CPTP map $\mathcal{E}\sim\{\hat{E}_k\}$, with $\hat{E}_k:\mathbb{C}^2\to \mathbb{C}^3$:
\begin{equation}
    \hat{E}_1 = c
    \begin{bmatrix}
        a&0\\
        0&b\\
        0&0\\
    \end{bmatrix},\quad
    \hat{E}_2 = c
    \begin{bmatrix}
        0&a\\
        0&0\\
        b&0\\
    \end{bmatrix}\;,
\end{equation}
where $a,b\in \mathbb{R}^+$ and the normalization factor is $c:=(a^2+b^2)^{-\frac{1}{2}}$.
Accordingly, the QEC matrix and its relevant partial trace are given by
\begin{equation}
    M = c^2
    \begin{bmatrix}
        a^2&&&a^2\\
        &b^2&&\\
        &&b^2&\\
        a^2&&&a^2
    \end{bmatrix}
    \;\;
    \Longrightarrow
    \;\;
    \sqrt{M}
    = c
    \begin{bmatrix}
        \frac{a}{\sqrt{2}}&&&\frac{a}{\sqrt{2}}\\
        &b&&\\
        &&b&\\
        \frac{a}{\sqrt{2}}&&&\frac{a}{\sqrt{2}}
    \end{bmatrix}
    \;\;
    \Longrightarrow
    \;\;
    \tr_\LL\sqrt{M} = c\left(\frac{a}{\sqrt{2}} + b\right) I_2\;\;,
\end{equation}
which must yield $[M,I_\LL\otimes \tr_\LL\sqrt{M}] = 0$.
As an optimal recovery map, the TC can be represented by
\begin{equation}
    \hat{R}^{\text{TC}}_1=
    \begin{bmatrix}
        \frac{1}{\sqrt{2}}&0&0\\
        0&1&0
    \end{bmatrix},\quad
    \hat{R}^{\text{TC}}_2=
    \begin{bmatrix}
        0&0&1\\
        \frac{1}{\sqrt{2}}&0&0
    \end{bmatrix}\;.
\end{equation}
Hence, we have the optimized channel fidelity $F^{\mathrm{op}}_e=F^{\mathrm{TC}}_e=\frac{2}{4}\left[c\left(\frac{a}{\sqrt{2}} + b\right)\right]^2$.

Additionally, this simple example possesses an interesting complementary channel $\mathcal{E}^{\text{c}}$. Since $M=M^T$ is the Choi matrix for $\mathcal{E}^{\text{c}}$, $\mathcal{E}^{\text{c}}$ is an entanglement-breaking channel if and only if $M$ is separable, which occurs at $b \le a$. 
In particular, if $b = \sqrt{2} a$, we simultaneously have $[M,(I_\LL\otimes \tr_\LL\sqrt{M})] = 0$ and $[C_\mathcal{E},(I_\LL\otimes \tr_\LL\sqrt{C_\mathcal{E}})] = 0$, where $C_\mathcal{E}$ is the Choi matrix for $\mathcal{E}$. In this case, the TC of the complementary channel is also optimal.

\subsubsection{The Pauli Channel}
We consider the Pauli channel of a general \textit{qudit} state ($d\ge 2$), where $\LL = \HH = \mathbb{C}^d$. In other words, the errors consist of pure logical Pauli errors.
The Kraus representation of such a Pauli channel is
\begin{equation}
    \mathcal{E}\sim \left\{\sqrt{p_g} g: g\in \PP,p_g\ge0,\sum_g p_g = 1\right\}\;,
\end{equation}
where $\PP=\langle X,Z\rangle/U(1)$ is the generalized Pauli group (modulo global phase). The generators $Z$ and $X$ are represented by the following ``clock and shift matrices'':
\begin{equation*}
    Z_{ab} = e^{i\frac{2\pi a}{d}}\delta_{ab}\;,\quad
    X_{ab} = \delta_{a\oplus 1, b}\;,
    \quad a,b\in\{0,1,\cdots, d-1\}\;,
\end{equation*}
where $a\oplus 1 \equiv a + 1\text{ mod }d$\;.
To find the QEC matrix and its powers, we first define the following $d$-by-$|\PP|d$ matrix:
\begin{equation}\label{eq:E_S}
    E:=\left[\hat{E}_1, \hat{E}_2,\cdots,\hat{E}_{|\PP|}\right]=\left[
        \sqrt{p_{g_1}} g_1, \sqrt{p_{g_2}} g_2, \cdots, \sqrt{p_{g_{|\PP|}}} g_{|\PP|}
    \right],\quad g_i\in\PP.
\end{equation}
For instance, in the special case of $d=2$, where $\PP=\{I,X,Y,Z\}$ with $2\times 2$ Pauli matrices, we have
\begin{equation*}
    E = \left[
    \begin{array}{cc|cc|cc|cc}
        \sqrt{p_I}&0&0&\sqrt{p_X}&0&-i\sqrt{p_Y}&\sqrt{p_Z}&0 \\
       0&\sqrt{p_I}&\sqrt{p_X}&0&i\sqrt{p_Y}&0&0&-\sqrt{p_Z} \\
    \end{array}\right]\;.
\end{equation*}
Eq.~\eqref{eq:E_S} always satisfies $I_\HH = EE^\dagger$ due to $g_ig_i^\dagger = I_\LL$ and $\sum_g p_g = 1$. Thus, the QEC matrix $M = E^\dagger E$ happens to be a projector:
\begin{equation*}
    M^{p} = E^\dagger (EE^\dagger)^{p-1}E = E^\dagger E = M\;,\quad \forall p\in \mathbb{R}\;,
\end{equation*}
(For convenience, we have swapped the subsystems $\LL$ and $\KK$ that $M$ is acting on.)
According to Eq.~\eqref{eq:R_petz_penrose} and $\sigma = I_\LL/d$, we find that
\begin{equation}\label{eq:E_StoPsi}
    \left[\hat{R}^{\mathrm{TC}\dagger}_1,\hat{R}^{\mathrm{TC}\dagger}_2,\cdots,\hat{R}^{\mathrm{TC}\dagger}_{|\PP|}\right] = E M^{-\frac{1}{2}} = EM = EE^\dagger E = E\;.
\end{equation}
Therefore, the TC of $\mathcal{E}$ is exactly its conjugate
\begin{equation}
    \mathcal{R}^{\text{TC}} = \mathcal{E}^\dagger\;.
\end{equation}
Note that the character $\tr g_i = 0$ for all $g_i\in \PP\setminus\{I_\LL\}$, implying that the off-diagonal blocks of $\sqrt{M}$ are trace-less. Consequently, we obtain $\tr_\LL\sqrt{M} = \tr_\LL M = d\cdot \mathrm{diag}\left(p_{g_1},p_{g_2},\cdots,p_{g_{|\PP|}}\right)$.
Thus, $[M,I_\LL\otimes \tr_\LL\sqrt{M}]=0$ holds if
\begin{equation}\label{eq:pgPauli}
    p_{g} = \left\{
    \begin{array}{cc}
        1/|\mathcal{S}|, \quad & g\in\mathcal{S}\\
        0,\quad & g\notin\mathcal{S}
    \end{array}
    \right.\;,
\end{equation}
where $\mathcal{S}\subseteq\PP$ is any non-empty set. Accordingly, the optimal channel fidelity is given by
\begin{equation*}
    F^{\mathrm{op}}_e = F^{\mathrm{TC}}_e =
    \frac{1}{d^2}\left\|\tr_\LL\sqrt{M}\right\|^2_F=\frac{1}{d^2}\sum_{g\in \mathcal{S}} \left(\frac{d}{|\mathcal{S}|}\right)^2 = \frac{1}{|\mathcal{S}|}\;.
\end{equation*} 
In the above case, no recovery (i.e. trivially applying the identity map as $\RR$) is also the optimal recovery, since $F^{\text{id}}_{} = \frac{1}{d^2} |\tr(I_\LL/\sqrt{|\mathcal{S}|})|^2 = 1/|\mathcal{S}| = F^{\mathrm{op}}_e$.
We present some examples of Pauli channels satisfying Eq.~\eqref{eq:pgPauli} as follows:
\begin{itemize}
    \item  $\mathcal{E}$ is the dephasing channel, with $\mathcal{S}=\{ I, Z,\cdots, Z^s\}$ for any $1\le s \le d-1$.
    \item $\mathcal{E}$ is the completely depolarizing channel, which corresponds to $\mathcal{S}=\PP\Longrightarrow |\mathcal{S}| = d^{2n}$. 
\end{itemize}

\subsubsection{The Classical-to-Classical Channel}
\label{sec:c2cchn}
Given discrete sets $\XX$ and $\YY$,
the classical-to-classical~(C2C) channel is a quantum channel that maps a classical state $\sum_{x\in\XX} p_\XX(x)\ket{x}\!\bra{x}$ to another classical state $\sum_{y\in\YY} p_{\YY}(y)\ket{y}\!\bra{y}$.
The Kraus representation of the C2C channel is given by the expression~\cite{wilde2017quantum}:
\begin{equation}\label{eq:Kraus_C2C}
    \mathcal{E}^{\text{C2C}}\sim\left\{\sqrt{p_{\YY|\XX}(y|x)}\ket{y}\!\bra{x}: x\in \XX, y\in \YY\right\}
\end{equation}
where $p_{\YY|\XX}$ is the conditional probability distribution describing a classical channel $\XX\to \YY$, which is normalized such that $\sum_y p_{\YY|\XX}(y|x) = 1$ for any $x$.
In Eq.~\eqref{eq:Kraus_C2C}, the quantum states $\ket{x}$ and $\ket{y}$ represent the orthonormal bases for the input Hilbert space $\LL = \mathbb{C}^{|\XX|}$ and the output Hilbert space $\HH = \mathbb{C}^{|\YY|}$, respectively. 
Since there are $n_K = |\XX||\YY|$ independent Kraus operators, it follows that $\KK = \mathbb{C}^{|\XX||\YY|}$.

We find that the TC is optimal in the following cases of C2C channels:
\begin{itemize}
    \item If $p_{\YY|\XX}(y|x_1)p_{\YY|\XX}(y|x_2) \propto \delta_{x_1 x_2}$, for all $y\in \YY$ and $x_i\in \XX$, then $\mathcal{E}^{\text{C2C}}$ results in a diagonal $M$. If $y$ is viewed as the corrupted classical codeword of the logical input $x$, then this condition implies that the classical information is distinguishable (and thus perfectly recoverable) from the corrupted output. To find the optimal channel fidelity given by the TC, we observe that $\tr_\LL\sqrt{M}$ is an $n_K$-by-$n_K$ diagonal matrix:
    \begin{equation}
        \tr_\LL\sqrt{M} = \mathrm{diag}\left(\sqrt{p_{\YY|\XX}(y_1|x_1)},\sqrt{p_{\YY|\XX}(y_2|x_1)},\cdots,\sqrt{p_{\YY|\XX}(y_{|\YY|}|x_{|\XX|})}\right)\;,
    \end{equation}
    so we have
    \begin{align}
            F^{\mathrm{TC}}_e &= \frac{1}{|\XX|^2}\left\|\tr_\LL\sqrt{M}\right\|_F^2 \label{eq:perfectEC_S_1}\\
             &= \frac{1}{|\XX|^2}\sum_{x,y}\left(\sqrt{p_{\YY|\XX}(y|x)}\right)^2 \label{eq:perfectEC_S_3}\\
             &= \frac{1}{|\XX|^2} \sum_{x} 1 
             = \frac{1}{|\XX|}\;, \label{eq:perfectEC_S_2}
    \end{align}
    where Eq.~\eqref{eq:perfectEC_S_1} follows from~\cite{zheng2024nearoptimal} and Eq.~\eqref{eq:perfectEC_S_2} uses the normalization condition $\sum_y p_{\YY|\XX}(y|x) = 1$.
    \item When $p_{\YY|\XX}(y|x) = p_{\YY}(y)$, implying that $y$ is independent of $x$. In this case, the channel maps an arbitrary input state to a fixed classical state $\sum_y p_{\YY}(y)\ket{y}\!\bra{y}$. Since $\ket{y}$ is orthogonal in the Kraus operator in Eq.~\eqref{eq:Kraus_C2C}, the QEC matrix $M$ is block-diagonal: 
    \begin{equation}
        M = \bigoplus_{y\in\YY} p_{\YY}(y)\mathcal{B}\;,
    \end{equation}
    where $\mathcal{B}$ is an $|\XX|^2$-by-$|\XX|^2$ matrix independent of $y$, since $\XX$ is decoupled from $\YY$. Specifically, $(\mathcal{B})_{[x_1x_2],[x'_1x'_2]} = \delta_{x_1x_2}\delta_{x'_1x'_2}$, so $\sqrt{\mathcal{B}} = |\XX|^{-\frac{1}{2}}\mathcal{B}\Rightarrow\tr_\LL\sqrt{\mathcal{B}}=|\XX|^{-\frac{1}{2}} I_{|\XX|}$. Due to the blockwise structure of $M$, we have $[M,I_\LL\otimes \tr_\LL\sqrt{M}]=0$. Thus, the optimal channel fidelity in this case is given by 
    \begin{equation}
        \begin{aligned}
            F^{\mathrm{TC}}_e &= \frac{1}{|\XX|^2}\left\|\tr_\LL\sqrt{M}\right\|_F^2 = \frac{1}{|\XX|^2}\bigg\|\bigoplus_{y\in\YY} \sqrt{p_{\YY}(y)} |\XX|^{-\frac{1}{2}} I_{|\XX|}\bigg\|_F^2 \\ 
            &=\frac{1}{|\XX|^2}\sum_{x,y}\left(\sqrt{p_{\YY}(y)}|\XX|^{-\frac{1}{2}}\right)^2 = \frac{1}{|\XX|}\;.
        \end{aligned}
    \end{equation}
\end{itemize}

\section{The Quantum Transduction Example}
\begin{figure}[b]
    \centering
    \includegraphics[width=0.8\linewidth]{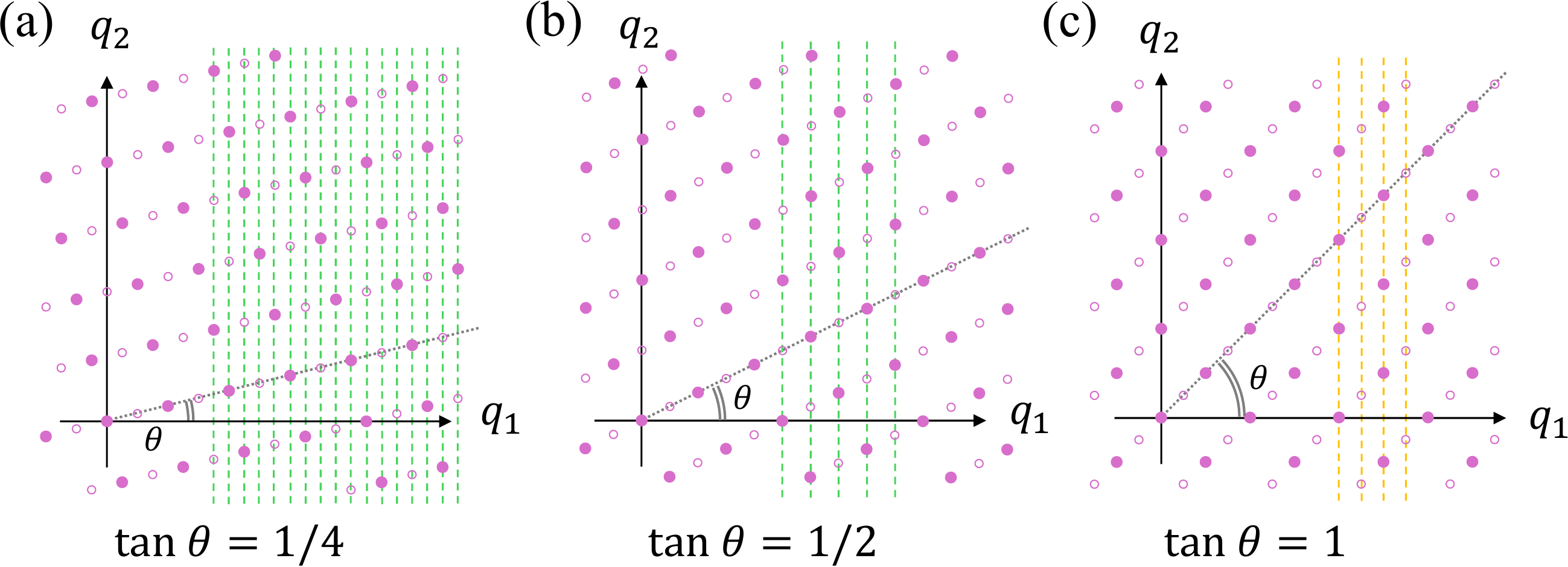}
    \caption{Each of these figures illustrates the 2D wavefunctions $\psi_{0;\Delta}(q_1,q_2)$ (filled circles) and $\psi_{1;\Delta}(q_1,q_2)$ (hollow circles) in the limit $\Delta \to 0$. The rotational angle $\theta = gt$ characterizes the interaction time between the two bosonic modes. In panels (a) and (b), the green dashed lines indicate the identical marginal distribution on the horizontal axis, given the lattice is infinite. For comparison, panel (c) displays distinguishable marginal distributions by yellow dashed lines. The values of transmissivity $\eta$ in these panels are, respectively, $\frac{1}{17}$, $\frac{1}{5}$ and $\frac{1}{2}$.}
    \label{fig:gkp_rotated_S}
\end{figure}
This section explains the details of the numerical test for the Gottesman-Kitaev-Preskill~(GKP) code used in the main text. The system of interest involves two bosonic modes $1$ and $2$, whose annihilation operators are $\hat{a}_1$ and $\hat{a}_2$, respectively. Their interaction is given by the unitary $U(t) = e^{-it\hat{H}}$, where $\hat{H} = -ig(\hat{a}^\dagger_1\hat{a}_2 - \hat{a}_1\hat{a}^\dagger_2)$ and $t$ is the interaction time. Such an interaction gives the transmissivity $\eta = \sin^2 gt$.
The two modes are initialized using the GKP encoding. To simplify our model, we consider only the square-lattice qubit GKP code~\cite{GKP2001}, where the logical space is defined by states $\ket{0_\Delta}$ and $\ket{1_\Delta}$, whose unnormalized wavefunction on the $q_i$-quadrature ($\hat{q}_i:=\frac{\hat{a}_i + \hat{a}_i^\dagger}{\sqrt{2}}$) are given as follows
\begin{align}
    \langle q_i \ket{0_\Delta} &\propto \sum_{s = -\infty}^{\infty} e^{-2\pi\Delta^2 s^2} e^{-(q_i - 2s\sqrt{\pi})^2/(2\Delta^2)}\;, \\
    \langle q_i \ket{1_\Delta} &\propto \sum_{s = -\infty}^{\infty} e^{-2\pi\Delta^2 s^2} e^{-(q_i - (2s+1)\sqrt{\pi})^2/(2\Delta^2)}\;.
\end{align}
The parameter $\Delta$ ($|\Delta|<1$) characterizes the expectation value of the energy of each mode as $\hbar\omega_i(\langle\hat{a}^\dagger_i\hat{a}_i\rangle + 1/2) \approx \hbar\omega_i/(2\Delta^2)$. When $\Delta \to 0$, both $\langle q_i\ket{0_\Delta}$ and $\langle q_i\ket{1_\Delta}$ approach the ``Dirac comb'' asymptotically with spacing $2\sqrt{\pi}$, with a relative shift of $\sqrt{\pi}$ along the $q_i$-quadrature.

In this work, we present the special case where system $1$ encodes a logical qubit as $\ket{\psi_\Delta}:=c_0 \ket{0_\Delta} + c_1 \ket{1_\Delta}$, and the state of system $2$ is initialized as $\ket{0_\Delta}$. Let $\psi_{\mu;\Delta}(q_1,q_2):=\bra{q_1,q_2}U(t)\ket{\mu_\Delta}\ket{0_\Delta}$, when $\Delta \to 0$, the functions $\psi_{0;{0}}$ and $\psi_{1;{0}}$ can be represented by two 2D Dirac combs on the $q_1$-$q_2$ plane, rotated by an angle $\theta = gt$.
As displayed in Fig.~\ref{fig:gkp_rotated_S}, the coordinates of the non-zero values of $\psi_{0;{0}}$ and $\psi_{1;{0}}$ are shown as filled and hollow circles, respectively. 
From the geometric relation in Fig.~\ref{fig:gkp_rotated_S}(a, b), it is clear that the marginal distributions of $\psi_{0;{0}}$ and $\psi_{1;{0}}$ on the $q_1$-quadrature are indistinguishable if $\tan\theta = \frac{2m_1+1}{2m_2}$ $(m_i\in\mathbb{Z}, m_2\ne 0)$. In other words, the reduced density matrix (marginal density matrix) on system $1$ after applying $U(t)$ does not reveal any logical information regarding $c_0$ and $c_1$, which implies the existence of perfect recovery of the encoded state from system $2$ when $\Delta \to 0$ and $\theta = \arctan \frac{2m_1+1}{2m_2}$.
In contrast, as shown in Fig.~\ref{fig:gkp_rotated_S}(c), the distinguishable marginal distributions of $\psi_{0;{0}}$ and $\psi_{1;{0}}$ on the $q_1$-quadrature reveal the ratio between $|c_0|^2$ and $|c_1|^2$, resulting in $F<1$ when recovering the encoded state from system $2$. (In fact, the symmetrical marginal distribution across both quadratures in Fig.~\ref{fig:gkp_rotated_S}(c) implies zero coherent information regardless of the logical state input.) 

The set of dashed lines in Fig.~\ref{fig:gkp_rotated_S} (in either panel (a),(b), or (c)) corresponds to the maximal set of independent Kraus operators $\hat{E}_k$ for the channel $\mathcal{E}_\eta$. They satisfy the condition $\hat{E}_k^\dagger \hat{E}_\ell=0$ if $k\ne \ell$, so $[M_\sigma, I_\LL\otimes\tr_\LL\sqrt{M}_\sigma] = 0$ for any $\sigma\in\Pos(\LL)$, which implies the optimality of the transpose channel (even if Fig.~\ref{fig:gkp_rotated_S}(c) gives $F_e<1$).

There would be more independent Kraus operators in a physical scenario where the energy is finite ($\Delta > 0$). The $\hat{E}_k^\dagger \hat{E}_\ell=0$ condition is slightly violated, depending on the width of the highly-squeezed peaks.
To obtain the numerical results in the main text, we approximate the joint continuous-variable system $\HH\otimes \KK = \mathbb{C}^{n}\otimes \mathbb{C}^{n}$ with a large $n\sim 10^3$. It requires $\mathcal{O}(n^3)$ steps to obtain $n_K<n$ dominant and independent Kraus operators (for instance, ignoring those Kraus operators with $\|\hat{E}_k\|_2\le 10^{-5}$). Subsequently, the QEC matrix $M$ is found with complexity $\mathcal{O}(n_K^2n)$, which can be used to reconstruct the effective Choi matrix $C'_{\mathcal{E}_\eta}$ of smaller size and to calculate $F^{\mathrm{TC}}_e$ and the commutator norm. The $F^{\mathrm{op}}_e$ is obtained by feeding $C'_{\mathcal{E}_\eta}$ into the algorithm in~\cite{Reimpell_Werner_2005PRL}, with approximately $n_{\text{iter}}\sim30$ iterations.

\end{document}